\documentclass[journal,cosmc]{IEEEtran}
\usepackage{amsfonts}
\usepackage{amssymb}
\usepackage{cite}
\usepackage[cmex10]{amsmath}
\usepackage{float}
\usepackage{color}
\usepackage{stfloats,fancyhdr}
\usepackage{amsmath}
\usepackage{algorithm}
\usepackage{algorithmic}
\usepackage{multirow}
\usepackage{changepage}
\usepackage[normalem]{ulem}
\usepackage{amsthm}

\newtheorem{lemma}{Lemma}

\newtheorem{remark}{Remark}

\IEEEoverridecommandlockouts

\ifCLASSINFOpdf
  \usepackage[pdftex]{graphicx}
  \DeclareGraphicsExtensions{.pdf,.jpeg,.png}
\else
  \usepackage[dvips]{graphicx}
  \DeclareGraphicsExtensions{.pdf}
\fi

\usepackage{subfigure}

\usepackage{fancybox,dashbox}

\begin{document}
%
\title{Minimizing Age of Information in Cognitive Radio-based IoT Systems: Underlay or Overlay?
\thanks{The work of C. Zhai was supported by the Fundamental Research Funds of Shandong University (2017TB0011), and the open research fund of National Mobile Communications Research Laboratory, Southeast University (2019D09). The work of Y. Li was supported by the Australian Research Council under Grant DP150104019 and DP190101988. The work of B. Vucetic was supported in part by the Australian Research Council Laureate Fellowship grant number FL160100032.}
\thanks{Y. Gu, Y. Li and B. Vucetic are with School of Electrical and Information Engineering, The University of Sydney, Sydney, NSW 2006, Australia (email: \{yifan.gu, yonghui.li, branka.vucetic\}@sydney.edu.au).}
\thanks{H. Chen is with Department of Information Engineering, The Chinese University of Hong Kong (e-mail: he.chen@ie.cuhk.edu.hk).
}
\thanks{C. Zhai is with the School of Information Science and Engineering, Shandong University, Qingdao 266237, China, and also with the National Mobile Communications Research Laboratory, Southeast University, Nanjing 210096, China (e-mail: chaozhai@sdu.edu.cn).}
}
\author{Yifan Gu, He Chen, Chao Zhai, Yonghui Li, and Branka Vucetic
}

\maketitle

\begin{abstract}
We consider a cognitive radio-based Internet-of-Things (CR-IoT) network consisting of one primary IoT (PIoT) system and one secondary IoT (SIoT) system. The IoT devices of both the PIoT and the SIoT respectively monitor one physical process and send randomly generated status updates to their associated access points (APs). The timeliness of the status updates is important as the systems are interested in the latest condition (e.g., temperature, speed and position) of the IoT device. In this context, two natural questions arise: (1) How to characterize the timeliness of the status updates in CR-IoT systems? (2) Which scheme, overlay or underlay, is better in terms of the timeliness of the status updates. To answer these two questions, we adopt a new performance metric, named the age of information (AoI). We analyze the average peak AoI of the PIoT and the SIoT for overlay and underlay schemes, respectively. Simple asymptotic expressions of the average peak AoI are also derived when the PIoT operates at high signal-to-noise ratio (SNR). Based on the asymptotic expressions, we characterize a critical generation rate of the PIoT system, which can determine the superiority of overlay and underlay schemes in terms of the average peak AoI of the SIoT. Numerical results validate the theoretical analysis and uncover that the overlay and underlay schemes can outperform each other in terms of the average peak AoI of the SIoT for different system setups.
\end{abstract}



\IEEEpeerreviewmaketitle

\section{Introduction}
Current wireless communication networks are mainly designed for human-operated terminals. The next generation of wireless communication networks will support the Internet of Things (IoT) technology that connects and controls autonomous machines. The next critical step in developing wireless networks will be to provide industry with ubiquitous ultra-reliable and low-latency communications for a large number of automated devices to bring transformative automation to industry, infrastructure and healthcare \cite{IoTbackground}. Due to the expected rapid growth in the number of IoT devices, one significant challenge of the IoT wireless networks is the scarcity of the radio spectrum. A promising solution to address the spectrum scarcity is the use of cognitive radio (CR) technology in the IoT, leading to CR-based IoT (CR-IoT). The CR-IoT has drawn considerable attention recently, see \cite{CR_IoT_Background1,CR_IoT_Background2} and references therein.

In the CR technology, a secondary system is allowed to access the spectrum formally allocated to a primary system, subject to little degradation introduced to the primary service. There are two fundamental access strategies for the secondary system, namely overlay and underlay schemes \cite{overunderoutage,outage2,outage3,overunder}. In the overlay scheme, the secondary system performs spectrum sensing and starts its transmission when the primary system is idle. In the underlay scheme, the secondary system can always access the spectrum subject to an interference constraint. However, it must satisfy the interference constraint at all times even if the primary system is idle.

Many emerging IoT services, such as smart cities, parking and environment monitoring, that could benefit from CR, will require status updates delivery \cite{CR_IoT_Background1,CR_IoT_Background2}. The status updates are the generated packets from the sensors which measure the status (e.g., temperature, humidity, position, and speed, etc.) of machines for automation and control purposes. Different from the conventional CR setups with the licensed spectrum being occupied by a less agile network without status updates delivery (e.g., TV and RADAR networks), we consider a contemporary scenario where both primary and secondary systems running agile IoT services with status updates delivery. This is motivated by the fact that more and more IoT services will run over licensed spectrum. For example, various IoT technologies over licensed cellular spectrum, e.g., enhanced machine type communication (eMTC) and narrow-band IoT (NB-IoT), have been standardized and deployed in practice \cite{SSIoT}. The primary IoT service may not always occupy the spectrum, and it provides an opportunity for the secondary IoT devices to sense the spectrum and transmit their data when needed.

The timeliness of the status updates is important for both primary and secondary systems running IoT services. For example, in smart parking applications, sensors measure the occupancy of a parking space and continuously report the measurements to the users. The timeliness of the generated status updates from the sensor is significant as the users are interested in the latest condition of the parking space, i.e., occupied or vacant. In fact, the conventional delay metric cannot adequately quantify the timeliness. The latter can be reflected by a recently introduced metric, termed the age of information (AoI) \cite{AoI_FCFS}. The AoI captures both the network delay and the generation time for each status update. If the most recently received status update carries the information sampled at time $r\left(t\right)$, the status update age at time $t$ can be defined as $t-r\left(t\right)$ \cite{AoI_FCFS}. Note that the delay can only reflect the timeliness at the instants when the status updates are successfully received, whereas the AoI measures the timeliness of information at the receiver at any given time. To our best knowledge, this is the first work that characterize the AoI performance of CR-IoT systems with age-sensitive traffics for both primary and secondary systems. 
\subsection{Related Work}
The outage probability of overlay and underlay schemes for different CR systems was studied in \cite{overunderoutage,outage2,outage3}. The authors in \cite{overunder} proposed a mixed strategy of overlay and underlay schemes to maximize the capacity of a CR system. Delay performance of CR systems for Poisson arrivals and bursty arrivals of packets was studied in \cite{CR_Delay}. \cite{CRIoT_sensing} focused on the energy efficiency, and \cite{CRIoT_throughput} considered the throughput maximization in designing spectrum sensing mechanisms for CR-IoT systems. Differently, we investigate the AoI performance of CR-IoT systems.

The AoI has attracted an upsurge of research interests very recently. The seminal work in \cite{AoI_FCFS} derived the average AoI of a First-Come-First-Serve (FCFS) queueing system for three different queueing models $M/M/1$, $M/D/1$, and $D/M/1$. It was shown in \cite{AoI_FCFS} that there exists an optimal generation rate for the status updates to minimize the average AoI, which is different from those that can maximize the throughput or minimize the delay. The authors in \cite{AoI_LCFS} extended \cite{AoI_FCFS} into a Last-Come-First-Serve (LCFS) queueing system and it was shown that compared with the FCFS, the LCFS can considerably reduce the average AoI by always transmitting the latest generated status updates. The authors in \cite{AEphremides_outoforder} investigated the average AoI of a more complex scenario, where the generated status updates may arrive at the destination out-of-order and the outdated status updates are obsolete. The average AoI of more complex queueing models, $M/G/1$ and $M/G/1/2$ were considered in \cite{LHuang_MG1} and \cite{AEphremides_MM12_deadline}, respectively. Due to the complexity and difficulty in analyzing the average AoI, Costa \textit{et al.} proposed a new performance metric, named the average peak AoI, which only counts the maximum values of the age achieved immediately prior to the reception of the status updates \cite{peakAoI}. It was concluded that the values of the average peak AoI and average AoI are closely related to each other.

In order to reduce the AoI and increase reliability, the classical automatic repeat request (ARQ) scheme was considered in \cite{AOIerror}, where the same information is retransmitted when it is erroneously received by the destination. The authors in \cite{YIFANAOI} studied the average AoI for the truncated ARQ by setting a maximum allowable transmission times for each status update. Hybrid ARQ (HARQ) schemes were analyzed in \cite{MG11HARQ,AoI_HARQ}. Seo and Choi introduced the AoI outage probability, which quantifies the probability that the peak AoI exceeds a threshold \cite{JSEO-Age-Outage}. It was demonstrated that the AoI outage probability is important for safety-critical applications. The update and transmitting strategies in an adversarial setting with jamming were studied in \cite{AoI_jamming1,AoI_jamming2}. In practice, multiple IoT devices may coexist in the same IoT network. The scheduling policies to minimize the AoI performance were studied in \cite{Ryates_multiplesources,R.Talak-Centralized-UnknownCSI-MAC,I.Kadota-Centralized-MAC,Z.Jiang-Decentralized,S.Kaul-Distributed-Centralized-MAC}. Finally, the AoI for energy harvesting systems were studied in \cite{AoI_EH1,AoI_EH2}.

There have been a few investigations into the AoI of CR networks \cite{AValehi_CR,SLeng_CR}. Valehi and Razi considered a CR system with framing that the generated packets can be aggregated to form a larger frame for transmission \cite{AValehi_CR}. The authors then optimized the energy efficiency of the secondary device by choosing the optimal number of aggregated packets under a given average AoI constraint. Leng and Yener investigated a CR system with an energy harvesting secondary device \cite{SLeng_CR}. The authors aimed to find optimal sensing and updating policy based on the energy availability and spectrum sensing results of the secondary user. The optimization problem was modelled by a partially observable Markov decision process. The major differences between our work and \cite{AValehi_CR,SLeng_CR} lie in two aspects. First of all, different from the FCFS model considered in \cite{AValehi_CR} and the ideal generate-at-will model considered in \cite{SLeng_CR}, we adopt the packet management policy such that the IoT devices can discard the stale status updates and serve new generated status updates, which is important for reducing the AoI. Furthermore, \cite{AValehi_CR} and \cite{SLeng_CR} only studied the case that the secondary system is age-sensitive. We consider a more challenging scenario that both primary and secondary systems are age-sensitive, where the tangled AoI evolution of both systems makes the theoretical analysis non-trivial.
\subsection{Contributions}
In this paper, we investigate the AoI performance of a CR-IoT system consisting of one primary IoT (PIoT) system and one secondary IoT (SIoT) system. The main contributions of this paper are summarized as follows:
\begin{itemize}
  \item We derive closed-form expressions of the average peak AoI for both the PIoT and the SIoT in an overlay scheme. In the overlay scheme, the SIoT can only access the spectrum when the PIoT is not transmitting and the evolution of the instantaneous AoI for the SIoT thus depends on the operation of the PIoT (i.e., the PIoT is transmitting or idle). This makes the analytical methods used in the existing papers no longer applicable and new methods are required.
  \item We derive closed-form expressions of the average peak AoI for the PIoT and the SIoT in an underlay scheme. In the underlay scheme, the PIoT and the SIoT may generate interference to each other and their AoI thus get entangled with each other. To tackle this challenge, we model the dynamic behaviors of the PIoT and the SIoT in the underlay scheme by a finite-state Markov Chain (MC).
  \item Simple asymptotic expressions of the average peak AoI for both the PIoT and the SIoT are also derived when the PIoT is operating at high signal-to-noise ratio (SNR). Based on the derived asymptotic expressions, we then characterize a critical generation rate of the PIoT which can determine whether the overlay or underlay scheme is more superior, in terms of the average peak AoI of the SIoT. Note that the average peak AoI of the PIoT is identical for both overlay and underlay schemes when the PIoT operates at high SNR. Simulation results validate our theoretical analysis, illustrate the PIoT performance degradation introduced by the underlay scheme, and show that the overlay and the underlay schemes may outperform each other in terms of the average peak AoI of the SIoT for different system setups. 
\end{itemize}

The rest of the paper is organized as follows. Section II introduces the system model, including the considered overlay and underlay schemes. The formal definition of the AoI, the peak AoI, and some other important intervals are also introduced. In Section III, we derive closed-form expressions for the average peak AoI of both the PIoT and the SIoT adopting the overlay scheme. In Section IV, we analyze the average peak AoI of the considered CR-IoT system using the underlay scheme, for both the PIoT and the SIoT. Section V provides the asymptotic analysis of the average peak AoI when the PIoT operates at high SNR. Numerical results are presented in Section VI to validate the theoretical analysis and compare the overlay and underlay schemes. Finally, conclusions are drawn in Section VII.

\textbf{\emph{Notation}}: Throughout this paper, we use $f_{X}(x)$ to denote the probability density function (PDF)/probability mass function (PMF) of a random variable $X$. $F_{X}(x)$ is the cumulative distribution function (CDF) of a random variable $X$. $Ei\left( {\cdot} \right)$ is the exponential integral function \cite[Eq. (8.21)]{Tableofintegral}. $\mathbb{E}\left[ A  \right]$ is the expectation of $A$ and $\mathbb{E}\left\{ {A\left| B \right.} \right\}$ represents the conditional expectation of $A$ under a given condition $B$. $P\left\{ {A\left| B \right.} \right\}$ is the conditional probability of $A$ under a given condition $B$. We use $\left( \cdot \right)^T$ to represent the transpose of a matrix or vector and ${\bf I}$ denotes the identity matrix.

\section{System Model}
\begin{figure}
\centering \scalebox{0.4}{\includegraphics{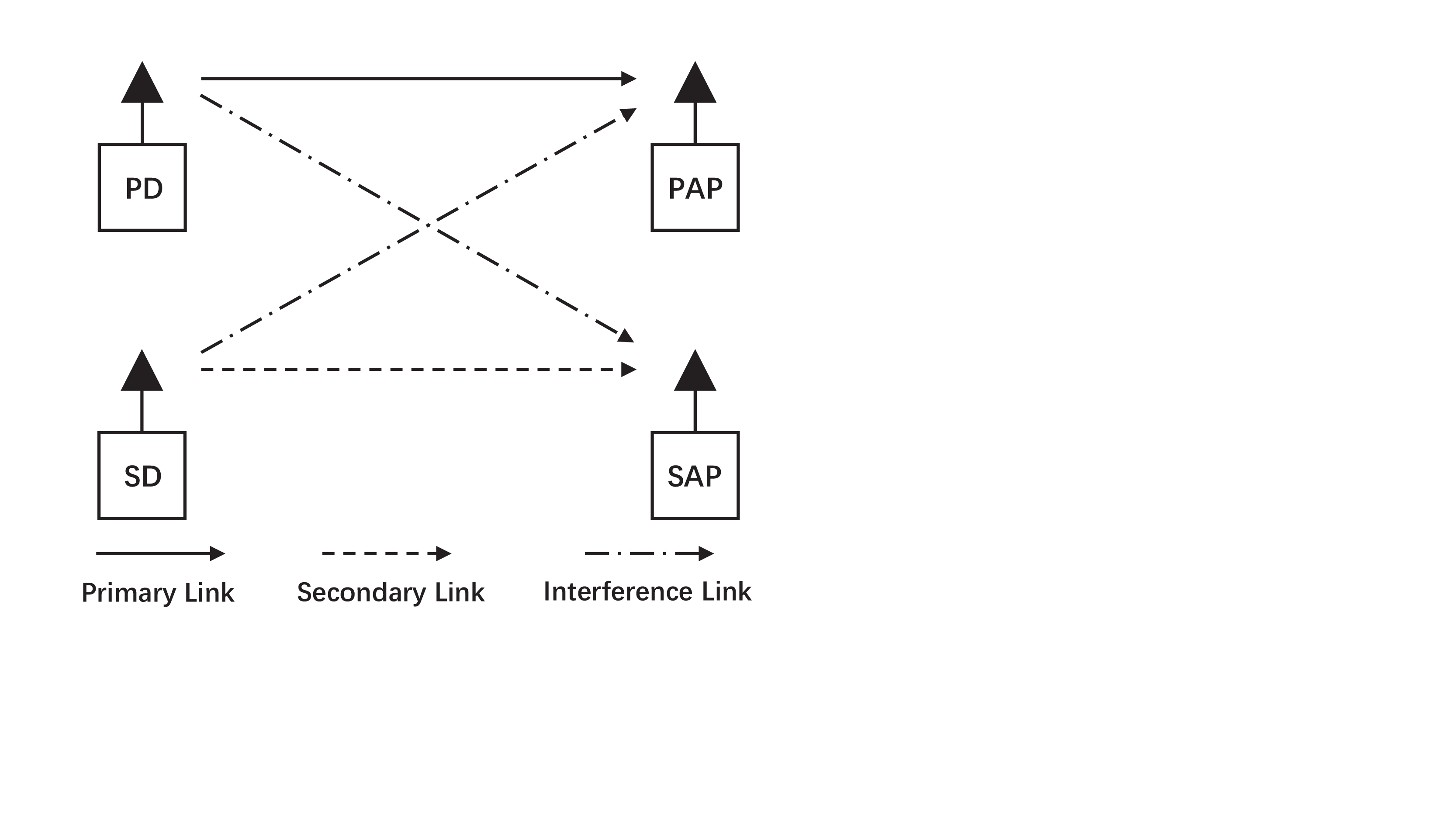}}
\caption{The considered cognitive radio-based IoT network. }\label{fig:systemmodel}
\end{figure}
We consider a CR-IoT network as depicted in Fig. \ref{fig:systemmodel}, which consists of one primary IoT device (PD), one primary access point (PAP), one secondary IoT device\footnote{The considered model can capture the IoT networks with multiple PDs and SDs working on orthogonal frequency channels, where there exists at most one SD in each channel to avoid causing severe interference to the PIoT. The scenario that multiple SDs coexist in each channel causing aggregate interference to the PIoT has been left as a future work.} (SD), and one secondary access point (SAP). We assume that the PD and the SD have their own status updates to be transmitted to their associated access points, the PAP and the SAP, respectively. Besides, the spectrum in the considered CR-IoT has only been licensed to the PIoT, and it is assumed that the PIoT is not aware of the existence of the SIoT. Thus, the SIoT can only acquire access to the spectrum when the PIoT is idle, or by adjusting the transmit power of the SD to strictly guarantee the interference constraint at the PAP.

We consider a slotted communication scenario where time is divided into slots with equal durations. Without loss of generality, we assume that the duration of each slot is normalized to unity. Moreover, all the links in the considered CR-IoT are assumed to undergo Rayleigh fading that the transmission of each status update of the PD and the SD occupies one time slot. The channel coefficients of all the links remain unchanged during each slot and change independently from one slot to another.

Errors may occur during the transmission of each status update and retransmission is performed to reduce the AoI for both the PIoT and the SIoT. For implementation simplicity, we adopt the classical ARQ and assume that the PAP and the SAP do not have a buffer and cannot combine the signals received from previous transmissions. Besides, the preemption of new status updates is considered for both the PIoT and the SIoT to reduce the AoI. Specifically, each round of retransmission contains the same information, and the PD (SD) keeps transmitting the current status update to the PAP (SAP) until the PAP (SAP) successfully decodes the status update or a new status update is generated at the PD (SD). When a new status update is generated at the PD (SD) and the PIoT (SIoT) is busy serving a current status update, the PD (SD) discards the current status update and starts to serve the new one. We consider two fundamental access schemes for the considered CR-IoT, namely overlay and underlay schemes, and they are elaborated in the subsequent subsection.
\subsection{Overlay and Underlay Schemes}
In the overlay scheme, we assume that the SIoT can perform perfect sensing \cite{spectrumsensing}, i.e., the probability of false alarm and miss detection is zero, and the amount of time used for sensing is negligible compared to the whole slot. Moreover, the SD can only transmit its status update to the SAP when the PIoT is idle, i.e., the PD has no status update to be transmitted to the PAP. Although the PIoT and the SIoT adopt the same classical ARQ, the transmission of the SIoT may be further interrupted by the PIoT when a new status update is generated at the PD. If this happens, the SIoT suspends its transmission until the service of the PIoT is finished. Moreover, if a new status update is generated at the SD while the PD is transmitting, the SD cannot start its transmission until the PIoT finishes its service.

In the underlay scheme, it is assumed that the SD can acquire the channel state information (CSI) of the SD-PAP link\footnote{The CSI can be obtained at the SD by periodically listening to the pilot signal broadcast by the PAP. The pilot sequence used by the PAP can be acquired by the SD from a band manager, who coordinates the information sharing between the PIoT and the SIoT \cite{overunderoutage,outage2,outage3,overunder}.} and adjust its transmit power to strictly satisfy the instantaneous interference constraint $I_C$ at the PAP \cite{spectrumsensing}. By satisfying the interference constraint, the SD is allowed to transmit its own status updates at any time even when the PD is transmitting. Compared with the overlay scheme, the underlay scheme acquires more access to the spectrum at a cost of limited transmit power. Besides, both the PIoT and the SIoT may suffer from the interference generated by each other.

By adopting the classical ARQ, we next give the outage probabilities of the PIoT and the SIoT for each round of transmission using the overlay and underlay schemes, respectively. We let $P_P$ and $P_S$ denote the peak transmit power constraints at the PD and the SD, respectively. As these outage probabilities have been widely studied in the existing literature, we only give the final results for brevity.
\subsubsection{Overlay Scheme}
In the overlay scheme, the PD and the SD use the peak transmit power to transmit the status updates. Let $\Omega_{PP}$ and $\Omega_{SS}$ denote the average channel power gains between the PD and the PAP, the SD and the SAP, respectively. The outage probabilities of the PIoT and the SIoT for each round of transmission in the overlay scheme are given by
\begin{equation}\label{Poutage}
{{\varphi} _{O,P}} = 1 - \exp \left( { - {{{N_0}\sigma_P} \over {{P_P}{\Omega _{PP}}}}} \right),
\end{equation}
\begin{equation}\label{SoutageO}
{\varphi_{O,S}} = 1 - \exp \left( { - {{{N_0}\sigma_S} \over {{P_S}{\Omega _{SS}}}}} \right),
\end{equation}
where $N_0$ is the power of the additive white Gaussian noise at the receiver side, $\sigma_P=2^{R_P}-1$ and $\sigma_S=2^{R_S}-1$ are the outage thresholds, and $R_P,R_S$ are the transmission rates of the PIoT and the SIoT, respectively.
\subsubsection{Underlay Scheme}
In the underlay scheme, for the case that either the PD or the SD is transmitting, there is no interference and the outage probability of the PIoT is the same as the overlay scheme given in (\ref{Poutage}), i.e., ${\varphi_{U,P}}={{\varphi} _{O,P}}$.

For the SIoT, the SD needs to adjust its transmit power to satisfy the interference constraint $I_C$. Let $H_{SP}$ denote the instantaneous channel power gain between the SD and the PAP. By considering the peak transmit power and interference constraint, the transmit power of the SD is given by $P= \min \left\{ {I_C \over H_{SP}}, P_S \right\}$. According to \cite[Eq. (8)]{outage3}, without interference from the PIoT, the outage probability of the SIoT using the underlay scheme is given by
\begin{equation}\label{SoutageU}
\begin{split}
{\varphi_{U,S}} &= 1 - \left[ {1 - {{\exp \left( { - {{{I_C}} \over {{P_S}{\Omega _{SP}}}}} \right)} \over {{{{\Omega _{SS}}{I_C}} \over {{\Omega _{SP}}\sigma_S{N_0}}} + 1}}} \right]\exp \left( { - {{\sigma_S{N_0}} \over {{P_S}{\Omega _{SS}}}}} \right),
\end{split}
\end{equation}
where $\Omega_{SP}$ is the average channel power gain between the SD and the PAP.

For the other case that both the PD and the SD transmit at the same slot, the PIoT and the SIoT generates interference to each other. Due to the peak transmit power constraint of the SD, the interference at the PAP can be expressed as $I_P= \min \left\{ I_C, P_S H_{SP} \right\}$. Let $H_{PP}$ denote the instantaneous channel power gain for the PD-PAP link, the outage probability of the PIoT with interference can be calculated as
\begin{equation}\label{PoutageI}
\begin{split}
{{\mathord{\buildrel{\lower3pt\hbox{$\scriptscriptstyle\frown$}}\over
 {\varphi}  } }_{U,P}}&=\Pr \left\{ {{{{P_P}{H_{PP}}} \over {{I_P} + {N_0}}} < \sigma_P} \right\}\\
 &=\int_0^{{{{I_c}} \over {{P_s}}}} {{f_{{H_{SP}}}}\left( x \right)} {F_{{H_{PP}}}}\left( {{{\sigma_P\left( {{P_s}x + {N_0}} \right)} \over {{P_p}}}} \right)dx \\
 &\quad + \int_{{{{I_c}} \over {{P_s}}}}^\infty  {{f_{{H_{SP}}}}\left( x \right)} {F_{{H_{PP}}}}\left( {{{\sigma_P\left( {{I_C} + {N_0}} \right)} \over {{P_p}}}} \right)dx\\
 &=1 - {{\exp \left[ { - {{{I_C}} \over {{P_S}{\Omega _{SP}}}} - {{\left( {{I_C} + {N_0}} \right){\sigma_P}} \over {{P_P}{\Omega _{PP}}}}} \right] + {{\exp \left( { - {{{N_0}{\sigma_P}} \over {{P_P}{\Omega _{PP}}}}} \right)} \over {{P_S}{\Omega _{SP}}{\sigma_P}}}} \over {1 + {{{P_P}{\Omega _{PP}}} \over {{P_S}{\Omega _{SP}}{\sigma_P}}}}},
 \end{split}
\end{equation}
where ${{f_{{H_{SP}}}}\left( x \right)} = {{\exp \left( { - {x \over {{\Omega _{SP}}}}} \right)} \over {{\Omega _{SP}}}}$ and ${F_{{H_{PP}}}}\left( x \right)=1 - \exp \left( { - {x \over {{\Omega _{PP}}}}} \right)$ denote the PDF and CDF, respectively, and the two integrals in (\ref{PoutageI}) are solved with the help of \cite[Eq. (3.351.1)]{Tableofintegral}.

At last, the outage probability of the SIoT with interference can be expressed as \cite{outage2}
\begin{align}
{{{{\mathord{\buildrel{\lower3pt\hbox{$\scriptscriptstyle\frown$}}\over
 \varphi } }}}_{U,S}} & =1 - {{\exp \left( {{{{N_0}{\sigma_S}} \over {{P_S}{\Omega _{SS}}}}} \right)} \over {{\Omega _{PS}}}} \left[ {{{1 - \exp \left( { - {{{I_C}} \over {{P_S}{\Omega _{SP}}}}} \right)} \over {{1 \over {{\Omega _{PS}}}} + {{{\sigma_S}{P_P}} \over {{P_S}{\Omega _{SS}}}}}} }\right. \nonumber \\
 &\quad \left.{- {{{I_C}{\Omega _{SS}}\exp \left( \eta  \right)Ei\left( { - \eta } \right)} \over {{P_P}{\Omega _{PS}}\exp \left( {{{{I_C}} \over {{P_S}{\Omega _{SP}}}}} \right){\sigma_S}}}} \right], \label{SoutageUI}\\ \nonumber
 \end{align}
where $\Omega_{PS}$ is the average channel power gain between the PD and the SAP, and $\eta  = \left( {{{{I_C}{\Omega _{SS}}} \over {{\Omega _{SP}}}} + {\sigma_S}{N_0}} \right)\left( {{1 \over {{P_S}{\Omega _{SS}}}} + {1 \over {{P_P}{\Omega _{PS}}{\sigma_S}}}} \right)$.
\subsection{The Age of Information}\label{AoIdefinition}
\begin{figure}[htb]
\centering \scalebox{0.5}{\includegraphics{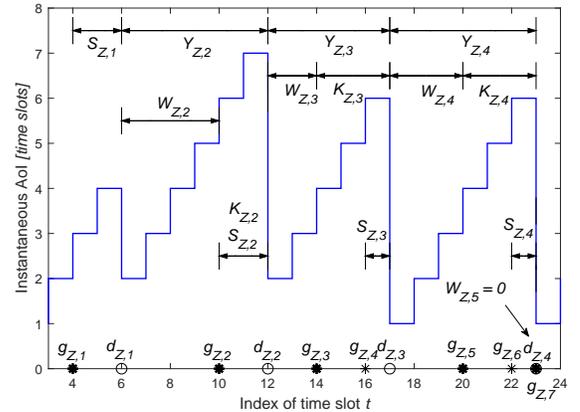}}
\caption{The instantaneous AoI for the PIoT (SIoT) using the overlay (underlay) scheme. }\label{fig:newaoi}
\end{figure}
As in \cite{Z.Jiang-Decentralized}, we assume that new status updates are generated at the PD and the SD according to an independent and identically distributed (i.i.d.) Bernoulli process with rates $p$ and $q$. Specifically, new status updates are generated at the PD and the SD with probabilities $p$ and $q$ at the beginning of each time slot, respectively. We now give some important definitions of the considered PIoT and the SIoT using underlay and overlay schemes. Note that the following definitions are used for both the overlay and underlay schemes. We let $g_{Z,j}$ denote the generation time of the $j$-th status update, where $Z \in \left\{P, S \right\}$ represents the PIoT and the SIoT, respectively. Recall that preemption is considered and a new generated status update preempts the current status update (i.e., the system discards the current status update and starts to serve the new one), the generated status updates may not always be successfully decoded at the access points due to preemption. We denote by ${d_{Z,i}}$ the departure time of the $i$-th status update that is successfully decoded at the PAP (SAP). For example, in Fig. \ref{fig:newaoi}, the status update generated at $g_{Z,3}$ is discarded due to preemption, and the third successfully received status update at ${d_{Z,3}}$ is the status update generated at $g_{Z,4}$. We let $S_{Z,i}$ be the service time of the $i$-th status update that is successfully received by the PAP (SAP). We use $r_{Z}\left(t\right)$ to represent the generation time of the most recently received status update at time slot $t$. The instantaneous AoI for the considered PIoT and SIoT at time slot $t$ is defined as
\begin{equation}
A_{Z}\left(t\right)=t-r_{Z}\left(t\right).
\end{equation}
In order to explain the above definition of AoI clearly, we depict an example for the instantaneous values of AoI in Fig. \ref{fig:newaoi} for 21 consecutive time slots. Since we consider a slotted communication scenario and the status updates can only be generated and received at discrete times, the instantaneous value of the AoI thus follows a staircase and changes only at the end of each slot.

We now give an example of instant AoI calculation by using Fig. \ref{fig:newaoi}. Here we explain how the AoI evolves starting from the 12-th time slot until the 17-th time slot. Specifically, the AoI drops to 2 in the 12-th time slot after a successful reception because the received status update is generated in the 10-th time slot. Because there is no status update generated until the 14-th time slot, the AoI grows 1 over each time slot, reaching a value of 4 in the 14-th time slot. After that, the AoI keeps increasing because the generated status update is not successfully delivered, reaching 6 in the 16-th time slot. In the 16-th time slot, a new status update is generated, and it is successfully delivered to the destination in the same time slot. The AoI is thus updated to the age of the received status update, e.g., 1, in the 17-th slot.

We follow \cite{AOIerror, peakAoI} and adopt the average peak AoI as the performance metric. Compared with the average AoI, the average peak AoI is more analytically tractable and it is closely related to the average AoI \cite{peakAoI}. The set of peak AoIs is the set of ``peak" values in the evolution of the instantaneous AoI, e.g., the $5,11,16,22$-th slots in Fig. \ref{fig:newaoi}. Mathematically, the $i$-th peak AoI of the PIoT and the SIoT is defined as
\begin{equation}\label{instantaneouspeakAoI}
{\Delta _{Z,i}} = {A_Z}\left( {{d_{Z,i}} - 1} \right).
\end{equation}
We now define some important intervals before calculating the average peak AoI of the PIoT and the SIoT.
\subsubsection{The interval $W_{Z,i}$}We denote $W_{Z,i}$ as the waiting time of the PIoT and the SIoT. $W_{Z,i}$ is defined as the time elapsed since the last successful reception of a status update, until the PD (SD) generates a new status update again. $W_{Z,i}$ can be expressed as
\begin{equation}
W_{Z,i}= \min \left\{ {{g_{Z,j}}\left| {{g_{Z,j}} \ge {d_{Z,i - 1}}} \right.} \right\} - {d_{Z,i - 1}}.
\end{equation}
Note that the waiting time $W_{Z,i}$ may be zero, corresponding to the case that the PIoT (SIoT) finishes the service of a status update in the current time slot and immediately generates a new status update at the beginning of the next time slot, e.g., $W_{Z,5}=0$ in Fig. \ref{fig:newaoi}.
\subsubsection{The interval $K_{Z,i}$}$K_{Z,i}$ is defined as the time elapsed since the first generated status update after the last successful reception, until the PAP (SAP) successfully decodes a status update again. $K_{Z,i}$ is formally defined as
\begin{equation}
K_{Z,i}={d_{Z,i}}- \min \left\{ {{g_{Z,j}}\left| {{g_{Z,j}} \ge {d_{Z,i - 1}}} \right.} \right\}.
\end{equation}
For a better understanding, we highlight the first generated status updates after the last successful reception in bold in Fig. \ref{fig:newaoi}. Note that the PIoT (SIoT) is idle during each interval $W_{Z,i}$ and busy, i.e., transmitting, during each interval $K_{Z,i}$. Besides, $K_{Z,i}$ is different from the service time $S_{Z,i}$ because preemption and interruption may happen during the interval $K_{Z,i}$, e.g., $K_{Z,3}>S_{Z,3}$ in Fig. \ref{fig:newaoi}. If no preemption or interruption happens, $K_{Z,i}$ and $S_{Z,i}$ are identical, e.g., $K_{Z,2}=S_{Z,2}$ as in Fig. \ref{fig:newaoi}.
\subsubsection{The interval $Y_{Z,i}$}The interdeparture interval between two successfully received status updates is given by
\begin{equation}
Y_{Z,i}={d_{Z,i}}-{d_{Z,i - 1}}.
\end{equation}
With the definition of $W_{Z,i}$ and $K_{Z,i}$, as shown in Fig. \ref{fig:newaoi}, $Y_{Z,i}$ can be alternatively expressed as
\begin{equation}
Y_{Z,i}=W_{Z,i}+K_{Z,i}.
\end{equation}
This result will be used in analyzing the average peak AoI in the next section. Based on the above definitions and (\ref{instantaneouspeakAoI}), the $i$-th peak AoI for the PIoT and the SIoT can be expressed as
\begin{equation}
{\Delta _{Z,i}} = {S_{Z,i - 1}}+{Y_{Z,i}}-1.
\end{equation}
For the considered CR-IoT, the sequences $\left\{S_{Z,1}, S_{Z,2},\cdots\right\}$ and $\left\{Y_{Z,1}, Y_{Z,2},\cdots\right\}$ form i.i.d. processes. Since we are interested in the average value of the peak AoI, we now drop the subscript index of the intervals. As such, the average peak AoI for the PIoT and the SIoT is given by
\begin{equation}\label{AoIexpression}
{{\bar \Delta }_Z}=\mathbb E \left[{S_{Z}}\right] +\mathbb E \left[{Y_{Z}}\right]-1.
\end{equation}

\section{AoI Analysis of the Overlay Scheme}
In this section, we analyze the average peak AoI of the PIoT and the SIoT for the overlay scheme by deriving the terms $\mathbb{E} \left[{S_Z}\right]$ and $\mathbb{E} \left[{Y_Z}\right]$ in (\ref{AoIexpression}). In the overlay scheme, the AoI performance of the SIoT depends on that of the PIoT because the SD cannot transmit when the PD is transmitting. On the other hand, the AoI performance of the PIoT is not influenced by the SIoT.
\subsection{The PIoT System}
\subsubsection{The evaluation of $\mathbb{E} \left[{Y_P}\right]$}
We first consider the expectation of $Y_P$. We notice that it is difficult to derive $\mathbb{E} \left[{Y_P}\right]$ directly due to its complicated behaviors that multiple status updates may be generated and preempted within each interval $Y_P$. To proceed, we leverage the fact that $Y_P=W_P+K_P$ and derive the expectation of $Y_P$ by calculating the expectations of $W_P$ and $K_P$. Since new status updates are generated at the PD according to a Bernoulli Process with a rate of $p$, the PMF of the waiting time $W_P$ is given by
\begin{equation}\label{PMFWP}
f_{W_P}\left( k \right) = \left(1-p\right)^{k}p.
\end{equation}
We can then obtain that $\mathbb{E} \left[{W_P}\right]={{1 - p} \over p}$.
According to the definition of $K_P$, we can deduce that $K_P$ is a group of consecutive information transmission slots, in which the PAP can only decode the status update correctly in the last slot. The PMF of $K_P$ is thus given by
\begin{equation}\label{PMFKP}
f_{K_P}\left( k \right) = {\varphi _{O,P}}^{k-1}\left(1-{\varphi _{O,P}}\right).
\end{equation}
From the PMF, we have $\mathbb{E} \left[{K_P}\right]={1 \over {1 - {\varphi _{O,P}}}}$, and the expectation of $Y_P$ is given by
\begin{equation}\label{YP}
\mathbb{E} \left[{Y_P}\right]=\mathbb{E} \left[{W_P}\right]+\mathbb{E} \left[{K_P}\right]={{1 - p} \over p} + {1 \over {1 - {\varphi _{O,P}}}}.
\end{equation}
\subsubsection{The evaluation of $\mathbb{E} \left[{S_P}\right]$}
We now derive the expectation of the service time $\mathbb{E} \left[{S_P}\right]$ for the PIoT system. The PMF of the service time $S_P$ is given in the following lemma.
\begin{lemma}\label{lemma1}
The PMF of the service time $S_P$ is given by
\begin{equation}\label{PMFSP}
\begin{split}
{f_{{S_P}}}\left( k \right)= {\varphi _{O,P}}^{k - 1}{\left( {1 - p} \right)^{k - 1}}\left( {1 - {\varphi _{O,P}} + p{\varphi _{O,P}}} \right).
\end{split}
\end{equation}
\end{lemma}
\begin{proof}
See Appendix \ref{app1}.
\end{proof}
With the PMF of $S_P$, $\mathbb{E} \left[{S_P}\right]$ can be derived as
\begin{equation}\label{SP}
\mathbb{E} \left[{S_P}\right]=\sum\limits_{k = 1}^\infty  {{f_{{S_P}}}\left( k \right)k}  = {1 \over {1 - {\varphi _{O,P}} + p{\varphi _{O,P}}}},
\end{equation}
where we use \cite[Eq. (0.231.2)]{Tableofintegral} to simplify the summation term.

Substituting (\ref{SP}) and (\ref{YP}) into (\ref{AoIexpression}), after some mathematical manipulations, the average peak AoI of the PIoT using the overlay scheme is obtained as
\begin{equation}\label{AoIOP}
{{\bar \Delta }_P}={{{\varphi _{O,P}} - p{\varphi _{O,P}}} \over {1 - {\varphi _{O,P}} + p{\varphi _{O,P}}}} + {{1 - {\varphi _{O,P}} + p{\varphi _{O,P}}} \over {p - p{\varphi _{O,P}}}}.
\end{equation}
\subsection{The SIoT system}
We now turn to the derivation of the average peak AoI for the SIoT by evaluating the terms $\mathbb{E} \left[{Y_S}\right]$ and $\mathbb{E} \left[{S_S}\right]$ in (\ref{AoIexpression}).
\subsubsection{The evaluation of $\mathbb{E} \left[{Y_S}\right]$}
Since $Y_S=W_S+K_S$, we also calculate the terms $\mathbb E \left[ W_S  \right]$ and $\mathbb E \left[ K_S  \right]$ in order to characterize $\mathbb E \left[ Y_S  \right]$. $W_S$ follows a Bernoulli Process with a rate of $q$, and the PMF of $W_S$ is given by
\begin{equation}\label{PMFWS}
f_{W_S}\left( k \right) = \left(1-q\right)^{k}q.
\end{equation}
With the PMF of $W_S$, we can obtain that $\mathbb E \left[ W_S  \right]={{1 - q} \over q}$.

Recall that the term $\mathbb E\left[K_P\right]$ for the PIoT is calculated by first deriving the PMF of $K_P$. However, in the overlay scheme, the SD has to suspend its transmission when the PIoT transmits its status updates. As a result, the SIoT may suspend its service several times during each interval $K_S$ and the duration of each suspension depends on whether new status updates are generated at the PIoT and the service time of the PIoT, which makes it complicated to derive the PMF of $K_S$. Inspired by \cite{AOIerror}, we use a recursive method to evaluate $\mathbb E\left[K_S\right]$, which does not require the PMF of $K_S$. Note that the considered case is much more complicated than the system considered in \cite{AOIerror}, because of the co-existence and different priorities between the PIoT and the SIoT. 

We notice that $K_S$ evolves differently for different initial slots. Specifically, if in the initial slot of $K_S$ the PIoT is idle, the SIoT can immediately start transmitting its status update. On the other hand, if in the initial slot of $K_S$ the PIoT is busy, the SIoT must wait until the service of the PIoT is finished before starting its own transmission. We now denote by $I_{K_S}$ and $B_{K_S}$ the events that the initial slot of $K_S$ is an idle slot or a busy slot, respectively. The expectation of $K_S$ for the SIoT using the overlay scheme can then be expressed as
\begin{equation}\label{overKS}
\mathbb E \left[ {K_S } \right] =\Pr\left\{I_{K_S}\right\}\mathbb E \left[ {K_S \left| I_{K_S} \right.} \right] + \Pr\left\{B_{K_S}\right\}\mathbb E \left[ {K_S \left| B_{K_S} \right.} \right].
\end{equation}
The two probabilities terms in (\ref{overKS}) are given in the following lemma.
\begin{lemma}\label{lemma2}
The probabilities that the initial slot of $K_S$ in the overlay scheme is an idle slot or a busy slot are given by
\begin{equation}\label{IK}
\Pr\left\{I_{K_S}\right\}={{\left( {1 - p} \right)\left( {1 - {\varphi _{O,P}} + {\varphi _{O,P}}q} \right)} \over {\left( {1 - p} \right)\left( {1 - {\varphi _{O,P}} + {\varphi _{O,P}}q} \right) + p}},
\end{equation}
\begin{equation}\label{BK}
\begin{split}
\Pr\left\{B_{K_S}\right\}&={p \over {\left( {1 - p} \right)\left( {1 - {\varphi _{O,P}} + {\varphi _{O,P}}q} \right) + p}}.
\end{split}
\end{equation}
\end{lemma}
\begin{proof}
See Appendix \ref{app2}.
\end{proof}
We now derive the two conditional expectations terms $\mathbb E \left[ {K_S \left| I_{K_S} \right.} \right]$ and $\mathbb E \left[ {K_S \left| B_{K_S} \right.} \right]$ in (\ref{overKS}). $\mathbb E \left[ {K_S \left| I_{K_S} \right.} \right]$ can be expressed as
\begin{align}
\mathbb E \left[ {K_S \left| I_{K_S} \right.} \right]&= \underbrace {\left( {1 - {\varphi _{O,S}}} \right)}_{T_1} +  \underbrace {{\varphi _{O,S}}\left( {1 - p} \right)\left( {1 + \mathbb E \left[ {K_S \left| I_{K_S} \right.} \right]} \right)}_{T_2} \nonumber \\
&\quad +  \underbrace {{\varphi _{O,S}}p\left( {1 + \mathbb E \left[ {K_S \left| B_{K_S} \right.} \right]} \right)}_{T_3}. \label{EKSIK}
\end{align}
When the initial slot of $K_S$ is an idle slot, the SIoT can transmit the generated status update immediately. The term $T_1$ in (\ref{EKSIK}) denotes the case that the status update is successfully decoded by the SAP after the initial transmission. The term $T_2$ in (\ref{EKSIK}) denotes the case that the first round of transmission contains an error and the PD does not generate a new status update in the next slot. In this case, the next slot of the SIoT is also an idle slot, and the conditional expectation of $K_S$ for this case can be expressed as $1+\mathbb E \left[ {K_S \left| I_{K_S} \right.} \right]$. The term $T_3$ in (\ref{EKSIK}) shows the case that the first round of transmission contains an error and the PD generates a new status update in the next slot. The next slot becomes a busy slot and the conditional expectation of $K_S$ is given by $1+\mathbb E \left[ {K_S \left| B_{K_S} \right.} \right]$. Similarly, $\mathbb E \left[ {K_S \left| B_{K_S} \right.} \right]$ can be expressed as
\begin{equation}\label{EKSBK}
\begin{split}
\mathbb E \left[ {K_S \left| B_{K_S} \right.} \right]&=\left( {1 - {\varphi _{O,P}}} \right)\left( {1 - p} \right)\left( {1 + \mathbb E \left[ {K_S \left| I_{K_S} \right.} \right]} \right)+ \\
&\quad \left[ {{\varphi _{O,P}}+ \left( {1 - {\varphi _{O,P}}} \right)p} \right]\left( {1 + \mathbb E \left[ {K_S \left| B_{K_S} \right.} \right]} \right).
\end{split}
\end{equation}
From (\ref{EKSIK}) and (\ref{EKSBK}), we can obtain that
\begin{equation}\label{KSIK}
\mathbb E \left[ {K_S \left| I_{K_S} \right.} \right]={{\left( {1 - {\varphi _{O,P}}} \right)\left( {1 - p} \right) + {\varphi _{O,S}}p} \over {\left( {1 - {\varphi _{O,P}}} \right)\left( {1 - {\varphi _{O,S}}} \right)\left( {1 - p} \right)}},
\end{equation}
\begin{equation}\label{KSBK}
\mathbb E \left[ {K_S \left| B_{K_S} \right.} \right]={{\left( {1 - {\varphi _{O,P}}} \right)\left( {1 - p} \right) + 1 - {\varphi _{O,S}} + {\varphi _{O,S}}p} \over {\left( {1 - {\varphi _{O,P}}} \right)\left( {1 - {\varphi _{O,S}}} \right)\left( {1 - p} \right)}}.
\end{equation}
Substituting the results derived in (\ref{IK}), (\ref{BK}), (\ref{KSIK}) and (\ref{KSBK}) into (\ref{overKS}), we can evaluate the expectation of $K_S$ and the term $\mathbb E \left[ {Y_S } \right]$ can be finally obtained in (\ref{YSO}) on top of the next page.
\begin{figure*}[!t]
\begin{equation}\label{YSO}
\begin{split}
\mathbb E \left[ {Y_S } \right]&=\mathbb E \left[ {W_S } \right]+\mathbb E \left[ {K_S } \right]={{1-q} \over {q}} +{{1} \over {1 - {\varphi _{O,S}}}} + {{p - {\varphi _{O,P}}{\varphi _{O,S}}p\left( {1 - p} \right)\left( {1 - q} \right)} \over {\left( {1 - {\varphi _{O,P}}} \right)\left( {1 - {\varphi _{O,S}}} \right)\left( {1 - p} \right)\left[ {\left( {1 - p} \right)\left( {1 - {\varphi _{O,P}} + {\varphi _{O,P}}q} \right) + p} \right]}}.
\end{split}
\end{equation}
\hrulefill
\vspace*{4pt}
\end{figure*}
\subsubsection{The evaluation of $\mathbb{E} \left[{S_S}\right]$}
We now turn to the derivation of the term $\mathbb E\left\{S_S\right\}$ in (\ref{AoIexpression}). $S_S$ also evolves differently when the initial slot of $S_S$ is an idle slot or a busy slot, respectively. We define $I_{S_S}$ and $B_{S_S}$ as the event that the initial slot of $S_S$ is an idle slot, or a busy slot, respectively. The expectation of $S_S$ can thus be evaluated as
\begin{equation}\label{overSS}
\mathbb E \left[ {S_S}\right]=\Pr \left\{ I_{S_S} \right\}\mathbb E \left[ {S_S \left| I_{S_S} \right.} \right]+\Pr \left\{ B_{S_S} \right\}\mathbb E \left[ {S_S \left| B_{S_S} \right.} \right].\\
\end{equation}
It is not easy to characterize the probability terms $\Pr\left\{I_{S_S}\right\}$ and $\Pr\left\{B_{S_S}\right\}$ in (\ref{overSS}). This is because the generated status update may be preempted by other status updates, and the events $I_{S_S}$ and $B_{S_S}$ only consider those status updates that are not discarded and finally decoded correctly by the SAP. In order to derive $\Pr\left\{I_{S_S}\right\}$ and $\Pr\left\{B_{S_S}\right\}$, we further define $I_S$ and $B_S$ as the events that the status update of the SIoT is generated in an idle slot or a busy slot, respectively. To be more exact, the events $I_S$ and $B_S$ are applied to all the generated status updates, including those being decoded correctly and those being preempted and discarded, while the events $I_{S_S}$ and $B_{S_S}$ are only applied to those status updates that are finally decoded correctly by the SAP. Furthermore, we use $\Phi $ to denote the event that the generated status update is not discarded due to preemption and finally decoded correctly at the SAP. With the above definitions, the probability terms $\Pr\left\{I_{S_S}\right\}$ and $\Pr\left\{B_{S_S}\right\}$ can be expressed as
\begin{equation}\label{ISSS}
\Pr \left\{ I_{S_S} \right\}={{\Pr \left\{ {\Phi \left| I_S \right.} \right\}\Pr \left\{ I_S \right\}} \over {\Pr \left\{ {\Phi \left| I_S \right.} \right\}\Pr \left\{ I_S \right\} + \Pr \left\{ {\Phi \left| B_S \right.} \right\}\Pr \left\{ B_S \right\}}},
\end{equation}
\begin{equation}\label{BSSS}
\Pr \left\{ B_{S_S} \right\}={{\Pr \left\{ {\Phi \left| B_S \right.} \right\}\Pr \left\{ B_S \right\}} \over {\Pr \left\{ {\Phi \left| I_S \right.} \right\}\Pr \left\{ I_S \right\} + \Pr \left\{ {\Phi \left| B_S \right.} \right\}\Pr \left\{ B_S \right\}}}.
\end{equation}
We now evaluate the probability terms $\Pr \left\{ I_S \right\}$, $\Pr \left\{ B_S \right\}$, $\Pr \left\{ {\Phi \left| I_S \right.} \right\}$ and $\Pr \left\{ {\Phi \left| B_S \right.} \right\}$ in order to characterize the two probability terms $\Pr \left\{ I_{S_S} \right\}$ and $\Pr \left\{ B_{S_S} \right\}$ in (\ref{overSS}). Recall that the PIoT keeps transmitting during each interval $K_P$ and remains idle during each interval $W_P$. Due to the fact that the generation of new status updates of the SIoT follows an i.i.d. process, the probability that the generation slot is an idle slot or a busy slot can be expressed as
\begin{equation}\label{PI}
\Pr \left\{ {{I_S}} \right\} = {{{W_p}} \over {{W_p} + {K_p}}} = 1 - {p \over {1 - {\varphi _{O,P}}  + p{\varphi _{O,P}} }},
\end{equation}
\begin{equation}\label{PB}
\Pr \left\{ {{B_S}} \right\} = {{{K_p}} \over {{W_p} + {K_p}}} = {p \over {1 - {\varphi _{O,P}}  + p{\varphi _{O,P}} }}.
\end{equation}
By using a similar recursive method as in (\ref{EKSIK}) and (\ref{EKSBK}), $\Pr \left\{ {\Phi \left| I_S \right.} \right\}$ and $\Pr \left\{ {\Phi \left| B_S \right.} \right\}$ can be expressed as
\begin{align}
\Pr \left\{ {\Phi \left| I_S \right.} \right\} & = \left( {1 - {\varphi _{O,S}}} \right) + {\varphi _{O,S}}\left( {1 - p} \right)\left( {1 - q} \right)\Pr \left\{ {\Phi \left| I_S \right.} \right\} \nonumber\\
&\quad +{\varphi _{O,S}}p\left( {1 - q} \right)\Pr \left\{ {\Phi \left| B_S \right.} \right\}, \label{PphiI}
\end{align}
\begin{equation}\label{PphiB}
\begin{split}
\Pr \left\{ {\Phi \left| B_S \right.} \right\} &= \left( {1 - {\varphi _{O,P}}} \right)\left( {1 - p} \right)\left( {1 - q} \right)\Pr \left\{ {\Phi \left| I_S \right.} \right\}+\\
&\quad \left[ {{\varphi _{O,P}} + \left( {1 - {\varphi _{O,P}}} \right)p} \right]\left( {1 - q} \right)\Pr \left\{ {\Phi \left| B_S \right.} \right\}.
\end{split}
\end{equation}
Note that the term $\left( 1-q \right)$ in (\ref{PphiI}) and (\ref{PphiB}) indicates the fact that the SD does not generate a new status update in the next slot and the considered status update is not preempted. Jointly considering (\ref{PphiI}) and (\ref{PphiB}), we can solve the terms $\Pr \left\{ {\Phi \left| I_S \right.} \right\}$ and $\Pr \left\{ {\Phi \left| B_S \right.} \right\}$, and they are given by
\begin{equation}\label{PhiI}
\begin{split}
\Pr \left\{ {\Phi \left| I_S \right.} \right\}= {{{q \over {\left( {1 - p} \right)\left( {1 - q} \right)}} + 1 - {\varphi _{O,P}}} \over {{q \over {\left( {1 - {\varphi _{O,S}}} \right)\left( {1 - p} \right)\left( {1 - q} \right)}} + 1 - {\varphi _{O,P}} - {{q{\varphi _{O,P}}{\varphi _{O,S}}} \over {1 - {\varphi _{O,S}}}}}},
\end{split}
\end{equation}
\begin{equation}\label{PhiB}
\begin{split}
&\Pr \left\{ {\Phi \left| B_S \right.} \right\}={{1 - {\varphi _{O,P}}} \over {{q \over {\left( {1 - {\varphi _{O,S}}} \right)\left( {1 - p} \right)\left( {1 - q} \right)}} + 1 - {\varphi _{O,P}} - {{q{\varphi _{O,P}}{\varphi _{O,S}}} \over {1 - {\varphi _{O,S}}}}}}.
\end{split}
\end{equation}
Substituting the above derived results in (\ref{PI}), (\ref{PB}), (\ref{PhiI}) and (\ref{PhiB}) to (\ref{ISSS}) and (\ref{BSSS}), the probabilities that the initial slot of $S_S$ is an idle slot or a busy slot are given by
\begin{equation}\label{IS}
\begin{split}
\Pr \left\{ I_{S_S} \right\}={{q + \left( {1 - {\varphi _{O,P}}} \right)\left( {1 - p} \right)\left( {1 - q} \right)} \over {q + \left( {1 - {\varphi _{O,P}}} \right)\left( {1 - p} \right)\left( {1 - q} \right) + p - pq}},\\
\end{split}
\end{equation}
\begin{equation}\label{BS}
\begin{split}
\Pr \left\{ B_{S_S} \right\}={{p - pq} \over {q + \left( {1 - {\varphi _{O,P}}} \right)\left( {1 - p} \right)\left( {1 - q} \right) + p - pq}}.\\
\end{split}
\end{equation}

Till now, we have derived the two probability terms $\Pr \left\{ I_{S_S} \right\}$ and $\Pr \left\{ B_{S_S} \right\}$ in (\ref{overSS}). We now derive the two conditional expectation terms in (\ref{overSS}) in order to finally characterize the expectation of $S_S$. Similar to (\ref{PphiI}) and (\ref{PphiB}), the conditional expectation of the service time under the condition that the initial slot of $S_S$ is an idle slot or a busy slot, can be calculated as
\begin{equation}\label{ESSI}
\begin{split}
\mathbb E \left[ {S_S \left| I_{S_S} \right.} \right]&={{1 - {\varphi _{O,S}}} \over {\Pr \left\{ {\Phi \left| I_S \right.} \right\}}} + \\
&\quad {\varphi _{O,S}}\left( {1 - p} \right)\left( {1 - q} \right)\left( {1 + \mathbb E \left[ {S_S \left| I_{S_S} \right.} \right]} \right) +\\
&\quad {{{\varphi _{O,S}}p\left( {1 - q} \right)\Pr \left\{ {\Phi \left| B_{S_S} \right.} \right\}} \over {\Pr \left\{ {\Phi \left| I_S \right.} \right\}}}\left( {1 + \mathbb E \left[ {S_S \left| B_{S_S} \right.} \right]} \right), \\
\end{split}
\end{equation}
\begin{equation}\label{ESSB}
\begin{split}
&\mathbb E \left[ {S_S \left| B_{S_S} \right.} \right]\\
&={{\left( {1 - {\varphi _{O,P}}} \right)\left( {1 - p} \right)\left( {1 - q} \right)\Pr \left\{ {\Phi \left| I_S \right.} \right\}} \over {\Pr \left\{ {\Phi \left| B_S \right.} \right\}}}\left( {1 + \mathbb E \left[ {S_S \left| I_{S_S} \right.} \right]} \right)\\
&\quad+\left( {{\varphi _{O,P}} + p - p{\varphi _{O,P}}} \right)\left( {1 - q} \right)\left( {1 + \mathbb E \left[ {S_S \left| B_{S_S} \right.} \right]} \right).
\end{split}
\end{equation}
In (\ref{ESSI}) and (\ref{ESSB}), each case in (\ref{PphiI}) and (\ref{PphiB}) are further conditioned on the events that the generated status update in an idle slot or in a busy slot is not preempted and finally decoded correctly by the SAP. The probability of these two events are derived in (\ref{PhiI}) and (\ref{PhiB}), respectively. The terms for each case in (\ref{PphiI}) and (\ref{PphiB}) thus should divide the total probabilities of these two events, i.e., $\Pr \left\{ {\Phi \left| I_S \right.} \right\}$ and $\Pr \left\{ {\Phi \left| B_S \right.} \right\}$, in deriving  (\ref{ESSI}) and (\ref{ESSB}), respectively. We can now solve the terms $\mathbb E \left[ {S_S \left| I_{S_S} \right.} \right]$ and $\mathbb E \left[ {S_S \left| B_{S_S} \right.} \right]$ from (\ref{ESSI}) and (\ref{ESSB}), and they are given by
\begin{equation}\label{SSI}
\mathbb E \left[ {S_S \left| I_{S_S} \right.} \right]={{1 + {{{\varphi _{O,S}}p\left( {1 - q} \right)\Pr \left\{ {\Phi \left| B \right.} \right\}} \over {\left[ {q + \left( {1 - {\varphi _{O,P}}} \right)\left( {1 - p} \right)\left( {1 - q} \right)} \right]\Pr \left\{ {\Phi \left| I \right.} \right\}}}} \over {1 - {\varphi _{O,S}}\left( {1 - q} \right)\left( {1 - p + p{{\Pr \left\{ {\Phi \left| B \right.} \right\}} \over {\Pr \left\{ {\Phi \left| I \right.} \right\}}}} \right)}},
\end{equation}
\begin{align}
\mathbb E \left[ {S_S \left| B_{S_S} \right.} \right]&={{1 + {{{\varphi _{O,S}}p\left( {1 - q} \right)\Pr \left\{ {\Phi \left| B \right.} \right\}} \over {\left[ {q + \left( {1 - {\varphi _{O,P}}} \right)\left( {1 - p} \right)\left( {1 - q} \right)} \right]\Pr \left\{ {\Phi \left| I \right.} \right\}}}} \over {1 - {\varphi _{O,S}}\left( {1 - q} \right)\left( {1 - p + p{{\Pr \left\{ {\Phi \left| B \right.} \right\}} \over {\Pr \left\{ {\Phi \left| I \right.} \right\}}}} \right)}} \nonumber\\
&\quad+{1 \over {q + \left( {1 - {\varphi _{O,P}}} \right)\left( {1 - p} \right)\left( {1 - q} \right)}}.\label{SSB}
\end{align}
With the results derived in (\ref{PhiI}), (\ref{PhiB}), (\ref{IS}), (\ref{BS}), (\ref{SSI}) and (\ref{SSB}), the expectation of the service time for the SIoT using the overlay scheme is finally obtained as
\begin{equation}\label{SS}
\begin{split}
&\mathbb E \left[ {S_S}\right]\\
&=\Pr \left\{ I_{S_S} \right\}\mathbb E \left[ {S_S \left| I_{S_S} \right.} \right]+\Pr \left\{ B_{S_S} \right\}\mathbb E \left[ {S_S \left| B_{S_S} \right.} \right]\\
&={{{q \over {1 - q}} + \left( {1 - {\varphi _{O,P}}} \right)\left( {1 - p} \right) + {{p\left[ {1 -  {\varphi _{O,P}}{\varphi _{O,S}}\left( {1 - p} \right)\left( {1 - q} \right)} \right]} \over {1 - {\varphi _{O,P}}\left( {1 - p} \right)\left( {1 - q} \right)}}} \over {{q \over {1 - q}} + \left( {1 - p} \right)\left[ {1 - {\varphi _{O,P}} - {\varphi _{O,S}}\left( {1 - {\varphi _{O,P}} + {\varphi _{O,P}}q} \right)} \right]}}.
\end{split}
\end{equation}
The average peak AoI of the SIoT using the overlay scheme can now be characterized in (\ref{AoIexpression}), where $\mathbb E \left[ {Y_S } \right]$ and $\mathbb E \left[ {S_S } \right]$ have been derived in (\ref{YSO}) and (\ref{SS}).

\section{AoI Analysis of the Underlay Scheme}
In this section, we analyze the average peak AoI of the PIoT and the SIoT for the underlay scheme. According to (\ref{AoIexpression}), we also derive the terms $\mathbb{E} \left[{S_Z}\right]$ and $\mathbb{E} \left[{Y_Z}\right]$ in order to characterize the average peak AoI. In the underlay scheme, the PIoT and the SIoT affect each other's performance due to the mutual interference. Because of the complicated behavior and the interdependence between the PIoT and the SIoT, the performance analysis methods used in Section III for the overlay scheme cannot be applied directly to the underlay scheme. We characterize the average peak AoI for the PIoT and the SIoT by applying a discrete Markov Chain (MC) to model the dynamic behaviors of the considered CR-IoT.
\subsection{The MC of the Underlay CR-IoT}
To proceed, we define the following four states of the MC. Specifically, state $s_1$ denotes that both the PIoT and the SIoT are idle, and they are waiting for new status updates to be generated. State $s_2$ denotes the case that the PIoT is busy while the SIoT is idle. In $s_2$, only the PIoT is transmitting and the SIoT is waiting for new status updates to be generated. State $s_3$ denotes the case that the SIoT is busy while the PIoT is idle. State $s_4$ denotes the case that both the PIoT and the SIoT are busy. The transition probability $T_{i,j}, \forall i,j \in \left\{1,2,3,4\right\}$ is defined as the probability of transition from state $s_i$ to state $s_j$. Based on the system description given in Section II, we can calculate the $4 \times 4$ transition matrix for the considered MC. In the following, for brevity, we only give the detailed analysis for the transition cases that start with $s_3$, which is complicated, and the transition cases starting with other states can be calculated similarly.
\subsubsection{Transition from $s_3$ to $s_1$} In this case, the PIoT remains idle during the transition, which means that no new status update is generated at the PIoT at the beginning of the next slot. The SIoT transits from busy to idle, it happens only when the status update of the SIoT is decoded correctly in the current slot and no new status update is generated at the SIoT at the beginning of the next slot. With the above analysis, the transition probability $T_{3,1}$ is given by
\begin{equation}
T_{3,1}={\left( {1 - p} \right)\left( {1 - {\varphi_{U,S}}} \right)\left( {1 - q} \right)}.
\end{equation}
\subsubsection{Transition from $s_3$ to $s_2$} In this transition case, new status update is generated at the PIoT such that it changes from idle to busy. The SIoT transits from busy to idle with a probability $\left( {1 - {\varphi_{U,S}}} \right)\left( {1 - q} \right)$ as analyzed before, and $T_{3,2}$ is given by
\begin{equation}
T_{3,2}={p\left( {1 - {\varphi_{U,S}}} \right)\left( {1 - q} \right)}.
\end{equation}
\subsubsection{Transition from $s_3$ to $s_3$} First of all, no status update is generated at the PIoT to keep the PIoT idle. Besides, two events may happen at the SIoT such that the SIoT can remain busy. The first event is that the status update cannot be decoded correctly in the current slot. The second event is that the status update is decoded successfully in the current slot but a new status update is generated at the beginning of the next slot. With the above analysis, the term $T_{3,3}$ can be evaluated as
\begin{align}
T_{3,3}&=\left( {1 - p} \right)\left[ {{\varphi _{U,S}}+\left( {1 - {\varphi _{U,S}}} \right)q} \right] \nonumber\\
&= \left( {1 - p} \right)\left( {q + {\varphi _{U,S}} - q{\varphi _{U,S}}} \right).
\end{align}
\subsubsection{Transition from $s_3$ to $s_4$} In this transition case, the PIoT changes from idle to busy while the SIoT remains busy, the transition probability can be derived as
\begin{equation}
T_{3,4}={p\left( {q + {\varphi_{U,S}} - q{\varphi_{U,S}}} \right)}.
\end{equation}
Till now, we have characterized the transition probabilities when the system starts with $s_3$. After deriving all the other transition probabilities, the $4\times 4$ transition matrix of the MC model is obtained in (\ref{MCtransition}) on top of the next page, where $\alpha=\left(1-{{\mathord{\buildrel{\lower3pt\hbox{$\scriptscriptstyle\frown$}}\over
 {\varphi}  } }_{U,P}}\right)\left(1-p\right)$, and $\beta=\left(1-{{{{\mathord{\buildrel{\lower3pt\hbox{$\scriptscriptstyle\frown$}}\over
 \varphi } }}}_{U,S}}\right)\left(1-q\right)$.
\begin{figure*}[!t]
\begin{equation}\label{MCtransition}
\pmb {M}=
\begin{bmatrix}
   &{\left( {1 - p} \right)\left( {1 - q} \right)} & {p\left( {1 - q} \right)} & {\left( {1 - p} \right)q} & {pq}  \\
   &{\left( {1 - {\varphi_{U,P}}} \right)\left( {1 - p} \right)\left( {1 - q} \right)} & {\left( {p + {\varphi_{U,P}} - p{\varphi_{U,P}}} \right)\left( {1 - q} \right)} & {\left( {1 - {\varphi_{U,P}}} \right)\left( {1 - p} \right)q} & {\left( {p +{\varphi_{U,P}} - p{\varphi_{U,P}}} \right)q}  \\
   &{\left( {1 - p} \right)\left( {1 - {\varphi_{U,S}}} \right)\left( {1 - q} \right)} & {p\left( {1 - {\varphi_{U,S}}} \right)\left( {1 - q} \right)} & {\left( {1 - p} \right)\left( {q + {\varphi_{U,S}} - q{\varphi_{U,S}}} \right)} & {p\left( {q + {\varphi_{U,S}} - q{\varphi_{U,S}}} \right)}  \\
   &{\alpha \beta} & {\left(1-\alpha\right)\beta} & {\alpha\left(1-\beta\right)} & {\left(1-\alpha\right)\left(1-\beta\right)}  \\
\end{bmatrix}
\end{equation}
\hrulefill
\end{figure*}

In order to analyze the average peak AoI of the PIoT and the SIoT using the underlay scheme, we are interested in the stationary distribution $\pmb \pi =\left\{\pi_1, \pi_2, \pi_3, \pi_4 \right\}$ of the MC, where $\pi_i$ denotes the stationary probability that the CR-IoT is at state $s_i$. To obtain the stationary distribution of the MC, we first prove that the considered MC is irreducible and row stochastic in the following lemma.
\begin{lemma}
The transition matrix $\pmb M$ given in (\ref{MCtransition}) of the MC model is irreducible and row stochastic for the generation rates $0<p,q<1$ and the outage probabilities ${\varphi_{O,P}}, {\varphi_{O,S}},{\varphi_{U,P}}, {\varphi_{U,S}}, {{\mathord{\buildrel{\lower3pt\hbox{$\scriptscriptstyle\frown$}}\over
 {\varphi}  } }_{U,P}}, {{{{\mathord{\buildrel{\lower3pt\hbox{$\scriptscriptstyle\frown$}}\over
 \varphi } }}}_{U,S}} \ne 1$.
\end{lemma}
\begin{proof}
We can verify that $\sum\limits_{j = 1}^4 {{T_{i,j}}}  = 1, \forall i$, which shows that the MC is row stochastic. Besides, since ${T_{i,j}} \ne 0, \forall i,j$ for $0<p,q<1$ and ${\varphi_{O,P}}, {\varphi_{O,S}},{\varphi_{U,P}}, {\varphi_{U,S}}, {{\mathord{\buildrel{\lower3pt\hbox{$\scriptscriptstyle\frown$}}\over
 {\varphi}  } }_{U,P}}, {{{{\mathord{\buildrel{\lower3pt\hbox{$\scriptscriptstyle\frown$}}\over
 \varphi } }}}_{U,S}} \ne 1$, all the states of the MC are connected to each other. The transition matrix is thus irreducible.
\end{proof}
Given that the transition matrix $\pmb M$ is irreducible and row stochastic, the stationary distribution of the MC can be derived by solving $\pmb \pi \pmb {M}= \pmb \pi $ and it is given by \cite{YIFANdistributedrelay}
\begin{equation}\label{MCstationary}
 \pmb {\pi} = {\left( {{\mathbf{M}}^{T} - \mathbf{I} + \mathbf{B} }\right)^{ - 1}}\mathbf{b},
\end{equation}
where ${\mathbf{B}_{i,j}}=1, \forall i, j$ and $\mathbf{b}={(1,1, \cdots ,1)^T}$. Based on the above derived stationary distribution of the MC, we are now ready to analyze the average peak AoI for the PIoT and the SIoT using the underlay scheme in the following subsection.
\subsection{Average Peak AoI Analysis}
In the underlay scheme, because the PIoT and the SIoT perform similarly, we derive the general terms $\mathbb E \left[{Y_{Z}}\right]$ and $\mathbb E \left[{S_{Z}}\right]$ in (\ref{AoIexpression}) in order to evaluate the average peak AoI for the PIoT and the SIoT, where $Z \in \left\{P, S \right\}$ represents the PIoT and the SIoT, respectively.
\subsubsection{The evaluation of $\mathbb{E} \left[{Y_Z}\right]$}
Due to the fact that $Y_Z = W_Z+K_Z$, the expectation of $Y_Z$ can be obtained by evaluating the expectations of $W_Z$ and $K_Z$. We first consider the terms for the PIoT, i.e., $W_P$ and $K_P$. Specifically, the PIoT is idle in the states $s_1$ and $s_3$, while busy in the states $s_2$ and $s_4$. From the definition of $W_P$ and $K_P$ given in Section II, we notice that the PIoT is idle during the interval $W_P$ and busy during the interval $K_P$. We can thus deduce that
\begin{equation}\label{equation1}
{\pi _1} + {\pi _3} = {{\mathbb{E}\left[W_P\right]} \over {\mathbb{E}\left[W_P\right]+\mathbb{E}\left[K_P\right]}},
\end{equation}
where $\pi_i$ is the stationary distribution of state $s_i$ of the MC. With the fact that ${\mathbb{E}\left[W_P\right]}={{1 - p} \over p}$, we can solve $\mathbb{E}\left[K_P\right]$ from (\ref{equation1}) and it is given by
\begin{equation}\label{KPU}
\mathbb{E}\left[K_P\right]={{\left( {1 - p} \right)\left( {{\pi _2} + {\pi _4}} \right)} \over {p\left( {{\pi _1} + {\pi _3}} \right)}}.
\end{equation}
The expectation of $Y_P$ can then be calculated as
\begin{equation}\label{YPU}
\mathbb{E}\left[Y_P\right]={\mathbb{E}\left[W_P\right]}+\mathbb{E}\left[K_P\right]={{ {1 - p} } \over {p\left( {{\pi _1} + {\pi _3}} \right)}}.
\end{equation}
Similarly, for the SIoT, the expectation of $Y_S$ can be derived as
\begin{equation}\label{YSU}
\mathbb{E}\left[Y_S\right]={{ {1 - q} } \over {q\left( {{\pi _1} + {\pi _2}} \right)}}.
\end{equation}

\subsubsection{The evaluation of $\mathbb{E} \left[{S_Z}\right]$}
We now turn to the derivation of ${\mathbb{E}\left[S_Z\right]}$ for the PIoT and the SIoT using the underlay scheme. We use $I_{S_Z}$ and $B_{S_Z}$ to denote the event that the initial slot of $S_Z$ is an idle slot or a busy slot, respectively, and the expectation of $S_Z$ is given by
\begin{equation}\label{SZ}
\mathbb E \left\{ {{S_Z}} \right\}=\Pr \left\{ I_{S_Z} \right\}\mathbb E \left\{ {S_Z \left| I_{S_Z} \right.} \right\}+\Pr \left\{ B_{S_Z} \right\}\mathbb E \left\{ {S_Z \left| B_{S_Z} \right.} \right\}.
\end{equation}
To characterize the two probability terms in (\ref{SZ}), i.e., $\Pr \left\{ I_{S_Z} \right\}$ and $\Pr \left\{ B_{S_Z} \right\}$, we further define $I_Z$ and $B_Z$ as the events that the status updates are generated in an idle slot or a busy slot, respectively. $\Phi $ represents the event that the generated status update is not discarded due to preemption and finally decoded correctly. Similar to (\ref{ISSS}) and (\ref{BSSS}), the probability that the initial slot of $S_Z$ is an idle slot or a busy slot can be expressed~as
\begin{equation}\label{ISZ}
\Pr \left\{ I_{S_Z} \right\} = {{\Pr \left\{ {{I_Z}} \right\}\Pr \left\{ {\Phi \left| {{I_Z}} \right.} \right\}} \over {\Pr \left\{ {{I_Z}} \right\}\Pr \left\{ {\Phi \left| {{I_Z}} \right.} \right\} + \Pr \left\{ {{B_Z}} \right\}\Pr \left\{ {\Phi \left| {{B_Z}} \right.} \right\}}},
\end{equation}
\begin{equation}\label{BSZ}
\Pr \left\{ B_{S_Z} \right\} = {{\Pr \left\{ {{B_Z}} \right\}\Pr \left\{ {\Phi \left| {{B_Z}} \right.} \right\}} \over {\Pr \left\{ {{I_Z}} \right\}\Pr \left\{ {\Phi \left| {{I_Z}} \right.} \right\} + \Pr \left\{ {{B_Z}} \right\}\Pr \left\{ {\Phi \left| {{B_Z}} \right.} \right\}}}.
\end{equation}
We now evaluate the probability terms $\Pr \left\{ {{I_Z}} \right\}$, $\Pr \left\{ {{B_Z}} \right\}$, $\Pr \left\{ {\Phi \left| {{I_Z}} \right.} \right\}$ and $\Pr \left\{ {\Phi \left| {{B_Z}} \right.} \right\}$ in order to derive the two probability terms $\Pr \left\{ I_{S_Z} \right\}$ and $\Pr \left\{ B_{S_Z} \right\}$ in (\ref{SZ}).
From the definition of the MC, the PIoT is idle during the states $s_1$ and $s_3$, while busy during the states $s_2$ and $s_4$. Besides, the SIoT is idle during the states $s_1$ and $s_2$, while busy during the states $s_3$ and $s_4$. The probabilities that a new status update of the PIoT (SIoT) is generated in an idle slot or a busy slot are thus given by
\begin{equation}\label{IPU}
\Pr \left\{ {{I_Z}} \right\} = \left\{ {\begin{matrix}
   {{\pi _1} + {\pi _2},\quad \text{if}\quad Z = P}  \\
   {{\pi _1} + {\pi _3},\quad \text{if}\quad Z = S}  \\
\end{matrix} } \right.,
\end{equation}
\begin{equation}\label{BPU}
\Pr \left\{ {{B_Z}} \right\} = \left\{ {\begin{matrix}
   {{\pi _3} + {\pi _4},\quad \text{if}\quad Z = P}  \\
   {{\pi _2} + {\pi _4},\quad \text{if}\quad Z = S}  \\
\end{matrix} } \right..
\end{equation}
We then use a similar recursive method as in Section III-B to calculate the other two probability terms, $\Pr \left\{ {\Phi \left| {{I_Z}} \right.} \right\}$ and $\Pr \left\{ {\Phi \left| {{B_Z}} \right.} \right\}$. For notation simplicity, we use $u_Z$ to denote the generation rate of $Z$ and $v_Z$ to denote the generation rate of the other IoT system, i.e., $u_P=p$, $v_P=q$ and $u_S=q$, $v_S=p$. We use $\xi_{Z}$ to denote the outage probability of $Z$ without interference and $\zeta_{Z}$ denote the outage probability of the other IoT system without interference, i.e., $\xi_{P}={\varphi _{U,P}}$ and $\zeta_{P}={\varphi _{U,S}}$. We let ${{\mathord{\buildrel{\lower3pt\hbox{$\scriptscriptstyle\frown$}}\over
 {\xi}  } }_{U,Z}}$ denote the outage probability of $Z$ with interference and ${{\mathord{\buildrel{\lower3pt\hbox{$\scriptscriptstyle\frown$}}\over
 {\zeta}  } }_{U,Z}}$ denote the outage probability of the other IoT system with interference, i.e., ${{{\mathord{\buildrel{\lower3pt\hbox{$\scriptscriptstyle\frown$}}\over
 {\xi}  } }_{U,P}}}={{\mathord{\buildrel{\lower3pt\hbox{$\scriptscriptstyle\frown$}}\over
 {\varphi}  } }_{U,P}}$ and ${{\mathord{\buildrel{\lower3pt\hbox{$\scriptscriptstyle\frown$}}\over
 {\zeta}  } }_{U,P}}={{\mathord{\buildrel{\lower3pt\hbox{$\scriptscriptstyle\frown$}}\over
 {\varphi}  } }_{U,S}}$. Similar to (\ref{PphiI}) and (\ref{PphiB}), the probability that the generated status update is finally decoded correctly when the generation slot is an idle slot or a busy slot can be expressed~as
\begin{align}
\Pr \left\{ {\Phi \left| I_Z \right.} \right\} & = \left( {1 - {\xi_{Z}}} \right) + {\xi_{U,Z}}\left( {1 - u_Z} \right)v_Z\Pr \left\{ {\Phi \left| B_Z \right.} \right\} \nonumber\\
&\quad +{\xi_{Z}}\left( {1 - u_Z} \right)\left( {1 - v_Z} \right)\Pr \left\{ {\Phi \left| I_Z \right.} \right\},\label{PphiI1}\\ \nonumber
\end{align}
\begin{equation}\label{PphiB1}
\begin{split}
\Pr \left\{ {\Phi \left| B_Z \right.} \right\} & = \left( {1 - {{\mathord{\buildrel{\lower3pt\hbox{$\scriptscriptstyle\frown$}}\over
 {\xi}  } }_{Z}}} \right) + \\
 &\quad {{\mathord{\buildrel{\lower3pt\hbox{$\scriptscriptstyle\frown$}}\over
 {\xi}  } }_{Z}}\left( {1 - u_Z} \right)\left[{{\mathord{\buildrel{\lower3pt\hbox{$\scriptscriptstyle\frown$}}\over
 {\zeta}  } }_Z}+\left(1-{{\mathord{\buildrel{\lower3pt\hbox{$\scriptscriptstyle\frown$}}\over
 {\zeta}  } }_Z}\right)v_Z\right]\Pr \left\{ {\Phi \left| B_Z \right.} \right\} \\
 &\quad +{{\mathord{\buildrel{\lower3pt\hbox{$\scriptscriptstyle\frown$}}\over
 {\xi}  } }_Z}\left(1-{{\mathord{\buildrel{\lower3pt\hbox{$\scriptscriptstyle\frown$}}\over
 {\zeta}  } }_Z}\right)\left( {1 - u_Z} \right)\left( {1 - v_Z} \right)\Pr \left\{ {\Phi \left| I_Z \right.} \right\}.
\end{split}
\end{equation}
We can solve $\Pr \left\{ {\Phi \left| I_Z \right.} \right\}$ and $\Pr \left\{ {\Phi \left| B_Z \right.} \right\}$ from (\ref{PphiI1}) and (\ref{PphiB1}), and they are given in (\ref{PhiIU}) and (\ref{PhiBU}) on top of the next page. With the results derived in (\ref{IPU}), (\ref{BPU}), (\ref{PhiIU}) and (\ref{PhiBU}), we have characterized the two probability terms $\Pr \left\{ I_{S_Z} \right\}$ and $\Pr \left\{ B_{S_Z} \right\}$ in (\ref{ISZ}) and (\ref{BSZ}), respectively. We next turn to the derivation of the two conditional expectation terms $\mathbb E \left\{ {S_Z \left| I_{S_Z} \right.} \right\}$ and $\mathbb E \left\{ {S_Z \left| B_{S_Z} \right.} \right\}$ in (\ref{SZ}) in order to finally obtain the expectation of service time of the PIoT (SIoT) for the underlay scheme.
\begin{figure*}[!t]
\begin{equation}\label{PhiIU}
\Pr \left\{ {\Phi \left| I_Z \right.} \right\}={{\left( {1 - {\xi _Z}} \right)\left[ {1 - {{\mathord{\buildrel{\lower3pt\hbox{$\scriptscriptstyle\frown$}}\over
 {\xi}  } }_Z}\left( {1 - u_Z} \right)\left( {v_Z + {{\mathord{\buildrel{\lower3pt\hbox{$\scriptscriptstyle\frown$}}\over
 {\zeta}  } }_Z} - v_Z{{\mathord{\buildrel{\lower3pt\hbox{$\scriptscriptstyle\frown$}}\over
 {\zeta}  } }_Z}} \right)} \right] + {\xi _Z}\left( {1 - {{\mathord{\buildrel{\lower3pt\hbox{$\scriptscriptstyle\frown$}}\over
 {\xi}  } }_Z}} \right)\left( {1 - u_Z} \right)v_Z} \over {\left[ {1 - {\xi _Z}\left( {1 - u_Z} \right)\left( {1 - v_Z} \right)} \right]\left[ {1 - {{\mathord{\buildrel{\lower3pt\hbox{$\scriptscriptstyle\frown$}}\over
 {\xi}  } }_Z}{{\mathord{\buildrel{\lower3pt\hbox{$\scriptscriptstyle\frown$}}\over
 {\zeta}  } }_Z}\left( {1 - u_Z} \right)} \right] - {{\mathord{\buildrel{\lower3pt\hbox{$\scriptscriptstyle\frown$}}\over
 {\xi}  } }_Z}\left( {1 - {{\mathord{\buildrel{\lower3pt\hbox{$\scriptscriptstyle\frown$}}\over
 {\zeta}  } }_Z}} \right)\left( {1 - u_Z} \right)v_Z}},
\end{equation}
\begin{equation}\label{PhiBU}
\Pr \left\{ {\Phi \left| B_Z \right.} \right\}={{\left( {1 - {{\mathord{\buildrel{\lower3pt\hbox{$\scriptscriptstyle\frown$}}\over
 {\xi}  } }_Z}} \right)\left[ {1 - {\xi _Z}\left( {1 - u_Z} \right)\left( {1 - v_Z} \right)} \right] + \left( {1 - {\xi _Z}} \right){{\mathord{\buildrel{\lower3pt\hbox{$\scriptscriptstyle\frown$}}\over
 {\xi}  } }_Z}\left( {1 - {{\mathord{\buildrel{\lower3pt\hbox{$\scriptscriptstyle\frown$}}\over
 {\zeta}  } }_Z}} \right)\left( {1 - u_Z} \right)\left( {1 - v_Z} \right)} \over {\left[ {1 - {\xi _Z}\left( {1 - u_Z} \right)\left( {1 - v_Z} \right)} \right]\left[ {1 - {{\mathord{\buildrel{\lower3pt\hbox{$\scriptscriptstyle\frown$}}\over
 {\xi}  } }_Z}{{\mathord{\buildrel{\lower3pt\hbox{$\scriptscriptstyle\frown$}}\over
 {\zeta}  } }_Z}\left( {1 - u_Z} \right)} \right] - {{\mathord{\buildrel{\lower3pt\hbox{$\scriptscriptstyle\frown$}}\over
 {\xi}  } }_Z}\left( {1 - {{\mathord{\buildrel{\lower3pt\hbox{$\scriptscriptstyle\frown$}}\over
 {\zeta}  } }_Z}} \right)\left( {1 - u_Z} \right)v_Z}}.
\end{equation}
\hrulefill
\vspace*{4pt}
\end{figure*}

According to (\ref{ESSI}) and (\ref{ESSB}), the conditional expectation of the service time $S_Z$, under the condition that its initial slot is an idle or a busy slot, can be expressed as
\begin{equation}\label{SPIPU}
\begin{split}
\mathbb E \left\{ {{S_Z}\left| {{I_{S_Z}}} \right.} \right\} &= {{\left( {1 - {\xi _Z}} \right)} \over {\Pr \left\{ {\Phi \left| {{I_Z}} \right.} \right\}}} + \\
&\quad {\xi _Z}\left( {1 - u_Z} \right)v_Z{{\Pr \left\{ {\Phi \left| {{B_Z}} \right.} \right\}} \over {\Pr \left\{ {\Phi \left| {{I_Z}} \right.} \right\}}}\left( {1 + \mathbb E \left\{ {{S_Z}\left| {{B_{S_Z}}} \right.} \right\}} \right) \\
&\quad + {\xi _Z}\left( {1 - u_Z} \right)\left( {1 - v_Z} \right)\left( {1 + \mathbb E \left\{ {{S_Z}\left| {{I_{S_Z}}} \right.} \right\}} \right),
\end{split}
\end{equation}
\begin{align}
&\mathbb E \left\{ {S_Z \left| B_{S_Z} \right.} \right\} \nonumber \\
&= {{\left( {1 - {{\mathord{\buildrel{\lower3pt\hbox{$\scriptscriptstyle\frown$}}\over
 {\xi}  } }_Z}} \right)} \over {\Pr \left\{ {\Phi \left| B_Z \right.} \right\}}} + {{\mathord{\buildrel{\lower3pt\hbox{$\scriptscriptstyle\frown$}}\over
 {\xi}  } }_Z}\left( {1 - u_Z} \right) \left[{{\mathord{\buildrel{\lower3pt\hbox{$\scriptscriptstyle\frown$}}\over
 {\zeta}  } }_Z}+\left(1-{{\mathord{\buildrel{\lower3pt\hbox{$\scriptscriptstyle\frown$}}\over
 {\zeta}  } }_Z}\right)v_Z\right] \nonumber\\
 &\quad \times \left(1+ \mathbb E \left\{ {S_Z \left| B_{S_Z} \right.} \right\} \right)+ {{\mathord{\buildrel{\lower3pt\hbox{$\scriptscriptstyle\frown$}}\over
 {\xi}  } }_Z}\left(1-{{\mathord{\buildrel{\lower3pt\hbox{$\scriptscriptstyle\frown$}}\over
 {\zeta}  } }_Z}\right)\left( {1 - u_Z} \right)\left( {1 - v_Z} \right) \nonumber\\
 & \quad \times {{\Pr \left\{ {\Phi \left| I_Z \right.} \right\}} \over {\Pr \left\{ {\Phi \left| B_Z \right.} \right\}}}\left(1+ \mathbb E \left\{ {S_Z \left| I_{S_Z} \right.} \right\} \right).\label{SPBPU}
\end{align}
From (\ref{SPIPU}) and (\ref{SPBPU}), we can solve $\mathbb E \left\{ {S_Z \left| I_{S_Z} \right.} \right\}$ and $\mathbb E \left\{ {S_Z \left| B_{S_Z} \right.} \right\}$ given in (\ref{SZIU}) and (\ref{SZBU}) on top of the next page. With the above derived terms, we have characterized the expectation of service time for the PIoT (SIoT).
\begin{figure*}[!t]
\begin{equation}\label{SZIU}
\mathbb E \left\{ {S_Z \left| I_{S_Z} \right.} \right\}={{1 - {{\mathord{\buildrel{\lower3pt\hbox{$\scriptscriptstyle\frown$}}\over
 {\xi}  } }_Z}\left( {1 - u_Z} \right)\left( {{{\mathord{\buildrel{\lower3pt\hbox{$\scriptscriptstyle\frown$}}\over
 {\zeta}  } }_Z} + v_Z - v_Z{{\mathord{\buildrel{\lower3pt\hbox{$\scriptscriptstyle\frown$}}\over
 {\zeta}  } }_Z}} \right) + {{\xi _Z}}\left( {1 - u_Z} \right)v_Z{{\Pr \left\{ {\Phi \left| {{B_Z}} \right.} \right\}} \over {\Pr \left\{ {\Phi \left| {{I_Z}} \right.} \right\}}}} \over {1 - {{\xi _Z}}\left( {1 - u_Z} \right)\left( {1 - v_Z} \right) - {{\mathord{\buildrel{\lower3pt\hbox{$\scriptscriptstyle\frown$}}\over
 {\xi}  } }_Z}\left( {1 - u_Z} \right)\left[ {{{\mathord{\buildrel{\lower3pt\hbox{$\scriptscriptstyle\frown$}}\over
 {\zeta}  } }_Z} + v_Z - v_Z{{\mathord{\buildrel{\lower3pt\hbox{$\scriptscriptstyle\frown$}}\over
 {\zeta}  } }_Z} - {{\xi_Z}}{{\mathord{\buildrel{\lower3pt\hbox{$\scriptscriptstyle\frown$}}\over
 {\zeta}  } }_Z}\left( {1 - u_Z} \right)\left( {1 - v_Z} \right)} \right]}},
\end{equation}
\begin{equation}\label{SZBU}
\mathbb E \left\{ {S_Z \left| B_{S_Z} \right.} \right\}={{1 - {{\xi_Z}}\left( {1 - u_Z} \right)\left( {1 - v_Z} \right) + {{\mathord{\buildrel{\lower3pt\hbox{$\scriptscriptstyle\frown$}}\over
 {\xi}  } }_Z}\left( {1 - {{\mathord{\buildrel{\lower3pt\hbox{$\scriptscriptstyle\frown$}}\over
 {\zeta}  } }_Z}} \right)\left( {1 - u_Z} \right)\left( {1 - v_Z} \right){{\Pr \left\{ {\Phi \left| {{I_Z}} \right.} \right\}} \over {\Pr \left\{ {\Phi \left| {{B_Z}} \right.} \right\}}}} \over {1 - {{\xi_Z}}\left( {1 - u_Z} \right)\left( {1 - v_Z} \right) - {{\mathord{\buildrel{\lower3pt\hbox{$\scriptscriptstyle\frown$}}\over
 {\xi}  } }_Z}\left( {1 - u_Z} \right)\left[ {{{\mathord{\buildrel{\lower3pt\hbox{$\scriptscriptstyle\frown$}}\over
 {\zeta}  } }_Z} + v_Z - v_Z{{\mathord{\buildrel{\lower3pt\hbox{$\scriptscriptstyle\frown$}}\over
 {\zeta}  } }_Z} - {{\xi_Z}}{{\mathord{\buildrel{\lower3pt\hbox{$\scriptscriptstyle\frown$}}\over
 {\zeta}  } }_Z}\left( {1 - u_Z} \right)\left( {1 - v_Z} \right)} \right]}},
\end{equation}

\hrulefill
\vspace*{4pt}
\end{figure*}

The average peak AoI of the PIoT and the SIoT adopting the underlay scheme can be finally characterized in (\ref{AoIexpression}), where the general terms $\mathbb E \left[{S_{Z}}\right]$ is given in (\ref{SZ}) and $\mathbb E \left[{Y_{Z}}\right]$ is derived in (\ref{YPU}) and (\ref{YSU}) for the PIoT and the SIoT, respectively.

\section{Asymptotic AoI Analysis}
The derived expressions for the average peak AoI in Section III and IV can be used to evaluate and compare the performances of overlay and underlay schemes. However, due to the complicated structures of the derived expressions, we cannot gain further insights for the impacts of various system parameters on the average peak AoI, and compare the overlay and underlay schemes mathematically. Motivated by this, in this section, we perform the asymptotic analysis of the average peak AoI of the PIoT and the SIoT when the PIoT operates at high SNR, i.e., the outage probability for each transmission of the PIoT approaches to zero.
\subsection{Overlay Scheme}
When the PIoT operates at high SNR and the term ${\varphi _{O,P}} \to 0$, from (\ref{AoIOP}), the average peak AoI of the PIoT can be simplified to
\begin{equation}\label{POasym}
{{\hat \Delta }_P} \approx {1 \over p}.
\end{equation}
For the SIoT, from (\ref{YSO}) and (\ref{SS}), the expectations of $Y_S$ and $S_S$ are simplified to
\begin{equation}
\mathbb E \left[ {Y_S}\right]\approx {{1 - q} \over q} + {1 \over {\left( {1 - {\varphi _{O,S}}} \right)\left( {1 - p} \right)}},
\end{equation}
\begin{equation}
\mathbb E \left[ {S_S}\right]\approx {1 \over {q + \left( {1 - p} \right)\left( {1 - q} \right)\left( {1 - {\varphi _{O,S}}} \right)}}.
\end{equation}
The average peak AoI for the SIoT when the PIoT operates at high SNR can be expressed as
\begin{align}
{{\hat \Delta }_S}&\approx {{1 - 2q} \over q} + {1 \over {\left( {1 - {\varphi _{O,S}}} \right)\left( {1 - p} \right)}} \nonumber\\
&\quad +{1 \over {q + \left( {1 - p} \right)\left( {1 - q} \right)\left( {1 - {\varphi _{O,S}}} \right)}}.\label{SOasym}
\end{align}
\begin{remark}
From (\ref{POasym}), we can see that the average peak AoI of the PIoT is inversely proportional to the generation rate $p$. This is intuitive because we consider a slotted scenario and the higher the generation rate, the lower the average peak AoI. From (\ref{SOasym}), we can see that the average peak AoI of the SIoT decreases as the outage probability ${\varphi _{O,S}}$ reduces and the generation rate $q$ increases. More importantly, the average peak AoI of the SIoT increases as the generation rate of the PIoT $p$ increases. This is understandable because the SIoT has less access to the spectrum as the generation rate of the PIoT increases in the overlay scheme.
\end{remark}
\subsection{Underlay Scheme}
We now perform the asymptotic analysis for the underlay scheme when the PIoT operates at high SNR and the outage probabilities ${\varphi_{U,P}},{{\mathord{\buildrel{\lower3pt\hbox{$\scriptscriptstyle\frown$}}\over
 {\varphi}  } }_{U,P}} \to 0$. First of all, it can be easily verified that at high SNR, the average peak AoI of the PIoT is also ${{\hat \Delta }_P}\approx{1 \over p}$, which is the same as that in the overlay scheme. This is because at high SNR, the PIoT can decode the status update easily even suffering from the interference of the SIoT. The AoI performance of the SIoT when the PIoT operates at high SNR is giving in the following lemma:
\begin{lemma}\label{lemma4}
When the PIoT operates at high SNR, the average peak AoI of the SIoT in the underlay scheme can be expressed as
\begin{equation}\label{SUasym}
\begin{split}
{{\hat \Delta }_S}& \approx{{1 - 2q} \over q}+{1 \over {\left( {1 - p} \right)\left( {1 - {\varphi _{U,S}}} \right) + p\left( {1 - {{{{\mathord{\buildrel{\lower3pt\hbox{$\scriptscriptstyle\frown$}}\over
 \varphi } }}}_{U,S}}} \right)}}+\\
 &\quad {1 \over {q + \left( {1 - q} \right)\left[{\left( {1 - p} \right)\left( {1 - {\varphi _{U,S}}} \right) + p\left( {1 - {{{{\mathord{\buildrel{\lower3pt\hbox{$\scriptscriptstyle\frown$}}\over
 \varphi } }}}_{U,S}}} \right)}\right]}}.
\end{split}
\end{equation}
\end{lemma}
\begin{proof}
See Appendix \ref{app3}.
\end{proof}
\begin{remark}
From (\ref{SUasym}), we can see that the average peak AoI of the SIoT gets smaller with the decrease of the outage probability term ${{{{\mathord{\buildrel{\lower3pt\hbox{$\scriptscriptstyle\frown$}}\over
 \varphi } }}}_{U,S}}$ as well as the increase of the generation rate $q$. Moreover, the average peak AoI of the SIoT becomes larger with the increase of the generation rate of the PIoT $p$.
\end{remark}
\subsection{Comparison of the Overlay and Underlay Schemes}
We now compare the performance of the overlay and underlay schemes in terms of the average peak AoI of the SIoT by comparing (\ref{SOasym}) and (\ref{SUasym}). Specifically, the difference between them lies in the terms ${\left( {1 - {\varphi _{O,S}}} \right)\left( {1 - p} \right)}$ and ${\left( {1 - p} \right)\left( {1 - {\varphi _{U,S}}} \right) + p\left( {1 - {{{{\mathord{\buildrel{\lower3pt\hbox{$\scriptscriptstyle\frown$}}\over
 \varphi } }}}_{U,S}}} \right)}$. Due to the fact that $\left( {1 - {\varphi _{U,S}}} \right)>\left( {1 - {{{{\mathord{\buildrel{\lower3pt\hbox{$\scriptscriptstyle\frown$}}\over
 \varphi } }}}_{U,S}}} \right)$, we can deduce that both terms are decreasing functions of $p$. Additionally, the limits of the two terms are $\mathop {\lim }\limits_{p \to 0} \left( {1 - {\varphi _{O,S}}} \right)\left( {1 - p} \right) = 1 - {\varphi _{O,S}}$, $\mathop {\lim }\limits_{p \to 1} \left( {1 - {\varphi _{O,S}}} \right)\left( {1 - p} \right) = 0$, $\mathop {\lim }\limits_{p \to 0} {\left( {1 - p} \right)\left( {1 - {\varphi _{U,S}}} \right) + p\left( {1 - {{{{\mathord{\buildrel{\lower3pt\hbox{$\scriptscriptstyle\frown$}}\over
 \varphi } }}}_{U,S}}} \right)} = {1 - {\varphi _{U,S}}}$ and $\mathop {\lim }\limits_{p \to 1} {\left( {1 - p} \right)\left( {1 - {\varphi _{U,S}}} \right) + p\left( {1 - {{{{\mathord{\buildrel{\lower3pt\hbox{$\scriptscriptstyle\frown$}}\over
 \varphi } }}}_{U,S}}} \right)} = {1 - {{{{\mathord{\buildrel{\lower3pt\hbox{$\scriptscriptstyle\frown$}}\over
 \varphi } }}}_{U,S}}}$. Moreover, recall that the SIoT is implemented with a peak transmit power constraint and the outage probability for each round of transmission of the overlay scheme is lower than that of the underlay scheme, i.e., ${\varphi _{O,S}} <{\varphi _{U,S}}$. With the above analysis, we can thus deduce that there exists a critical generation rate of the PIoT $0<p^*<1$ such that the average AoI of the SIoT is identical for both the overlay and underlay schemes. The critical generation rate is the solution to the equation ${\left( {1 - {\varphi _{O,S}}} \right)\left( {1 - p} \right)}={\left( {1 - p} \right)\left( {1 - {\varphi _{U,S}}} \right) + p\left( {1 - {{{{\mathord{\buildrel{\lower3pt\hbox{$\scriptscriptstyle\frown$}}\over
 \varphi } }}}_{U,S}}} \right)}$, and it is given by
 \begin{equation}
 p^*= {{{\varphi _{U,S}}-{\varphi _{O,S}}}\over{1-{{{{\mathord{\buildrel{\lower3pt\hbox{$\scriptscriptstyle\frown$}}\over
 \varphi } }}}_{U,S}}+{\varphi _{U,S}}-{\varphi _{O,S}}}}.
 \end{equation}
 More specifically, if the generation rate of the PIoT $p>p^*$, the underlay scheme outperforms the overlay scheme in terms of the average peak AoI of the SIoT and the overlay scheme is superior to the underlay scheme when $p<p^*$.

\section{Numerical and Simulation Results}
In this section, we present the numerical and simulation results of the considered overlay and underlay schemes in terms of the average peak AoI of the PIoT and the SIoT. In order to capture the effect of the path-loss, we set the average power gains of the desired links ${\Omega_{XY}} = {10^{-3} \over {1 + {d_{XY}^\omega}}}$, where $X,Y \in \left\{P,S\right\}$ represents the IoT devices and the dedicated access points of the PIoT and the SIoT, respectively. $d_{XY}$ denotes the distance between nodes $X$ and $Y$, and $\omega  \in \left[ {2,5} \right]$ is the path-loss factor. Note that a 30 dB average signal power attenuation is assumed at a reference distance of 1 meter (m) in the above channel model. In all the following simulations, we set the distances between the PD and the PAP, the SD and the SAP $d_{PP}=d_{SS}=100$ m, the path-loss factor $\omega  = 3$, the noise power $N_0=- 80$ dBm, and the transmission rates $R_P=R_S=1$ bps/Hz.
\begin{figure}
\centering
  {\scalebox{0.5}{\includegraphics {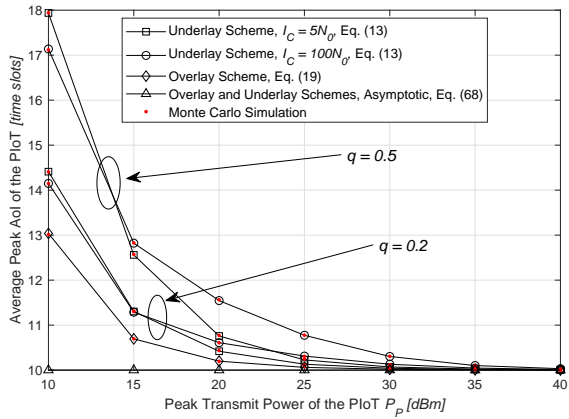}}}
\caption{The average peak AoI of the PIoT versus the peak transmit power of the PIoT for different generation rate $q$ and interference constraint $I_C$, where $P_S=25$ dBm, $d_{SP}=100$ m, $d_{PS}=150$ m and $p=0.1$.
\label{montsimuP1}}
\end{figure}
\begin{figure}
\centering
  {\scalebox{0.5}{\includegraphics {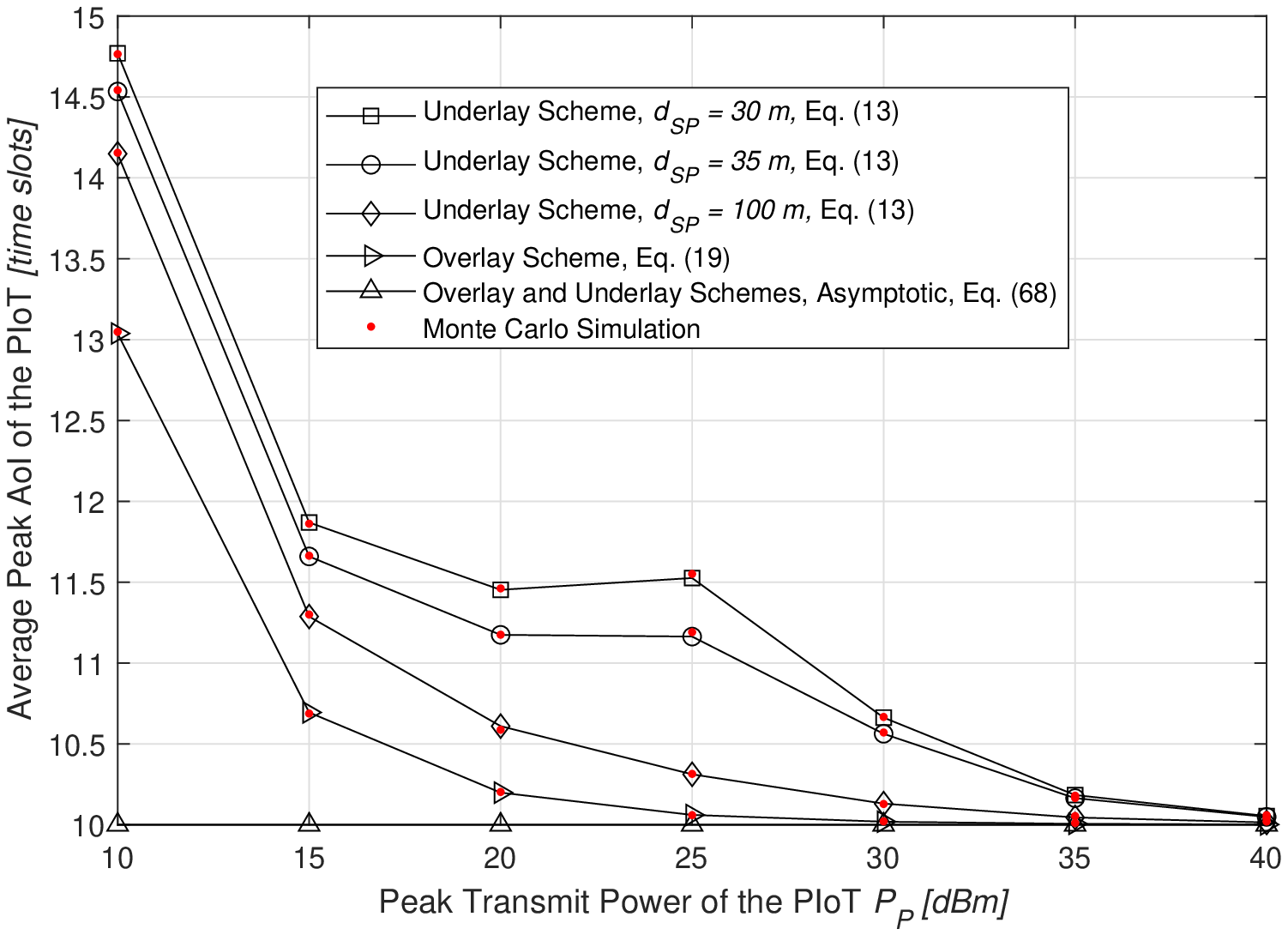}}}
\caption{The average peak AoI of the PIoT versus the peak transmit power of the PIoT for different distances between the SD and the PAP, where $P_S=25$ dBm, $d_{PS}=150$ m, $p=0.1$, $q = 0.2$, and $I_C = 100N_0$.
\label{montsimuP2}}
\end{figure}
\begin{figure}
\centering
  {\scalebox{0.6}{\includegraphics {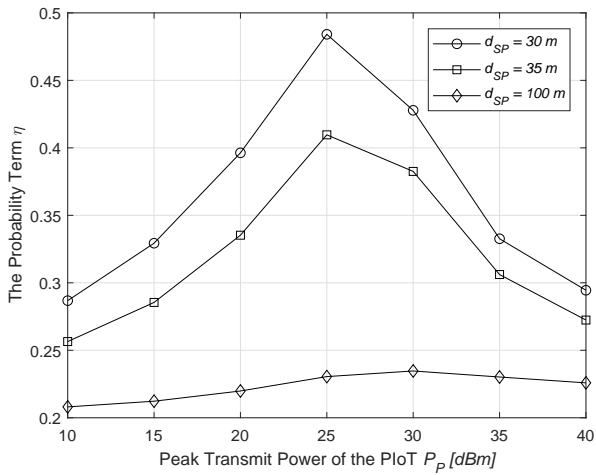}}}
\caption{The probability $\eta$ versus $P_P$ for the same settings as Fig. \ref{montsimuP2}.
\label{fig:review}}
\end{figure}
Fig. \ref{montsimuP1} and Fig. \ref{montsimuP2} depict the average peak AoI of the PIoT versus the peak transmit power of the PIoT for different system setups. From Fig. \ref{montsimuP1} and \ref{montsimuP2}, we first can see that both the derived analytical and asymptotic expressions of the average peak AoI for the PIoT coincide well with the Monte Carlo Simulation, which validates our theoretical analysis provided in Section III and IV for the PIoT. Besides, from Fig. \ref{montsimuP1} and Fig. \ref{montsimuP2}, we can observe that compared to the overlay scheme, the underlay scheme introduces degradation to the average peak AoI of the PIoT. The degradation generally decreases as the generation rate of the SIoT $q$ and the interference constraint $I_C$ decreases. It also decreases as the distance between the SD and the PAP $d_{SP}$ increases. This is understandable since the SIoT generates interference to the PIoT in the underlay scheme and the more the generated interference, the larger the degradation caused by the underlay scheme. Moreover, in Fig. \ref{montsimuP1}, we can find out that the increase of $I_C$ may improve the performance of the PIoT in the underlay scheme when the peak transmit power of the PIoT is low. This is because that an increasing of $I_C$ on one hand degrades the performance of the PIoT by generating more interference. On the other hand, the transmit power of the SIoT is also increased such that it can use fewer slots to finish the service of its generated status updates. Therefore, the PIoT now has more access to the spectrum without suffering from the interference of the SIoT. This is particularly important when the peak transmit power of the PIoT is low since it can hardly decode the status updates correctly when suffering from the interference of the SIoT.

In Fig. \ref{montsimuP2}, we can see that the average peak AoI of the PIoT for the overlay scheme decreases as the peak transmit power of the PIoT $P_P$ increases. This observation verifies the analysis given in Remark 1. We can also observe from Fig. \ref{montsimuP2} that there exists a non-convex phenomenon in the underlay scheme for $d_{SP} =30,35$ m. To explain this, we define a new parameter $\eta$ as the probability of the event that the PIoT suffers from interference when it is transmitting (i.e., the SIoT transmits at the same time). We clarify that increasing $P_P$ will influence both the outage probabilities and the value of $\eta$. The corresponding change of the AoI for the PIoT is the composition of these effects. On the one hand, it is straightforward that increasing $P_P$ will reduce the PIoT's outage probabilities for both cases with and without interference from the SIoT, which will further reduce the AoI of the PIoT. On the other hand, increasing $P_P$ will affect the value of $\eta$. However, the impact of $P_P$ on $\eta$ is not so straightforward due to the complicated interplay between the PIoT and the SIoT in the underlay scheme. To have a better understanding, we simulate the curves of $\eta$ versus $P_P$ in Fig. \ref{fig:review} for the same three settings as in Fig. \ref{montsimuP2}. We can observe from Fig \ref{fig:review} that for the cases $d_{SP}=30,35$ m, the value of $\eta$ first increases then decreases as $P_P$ increases, while for $d_{SP}=100$ m, the value of $\eta$ is relatively stable when $P_P$ changes. The increase of $\eta$ means that the PIoT has a higher probability to collide with the SIoT. This makes the PIoT's outage probability tends to equal to ${{\mathord{\buildrel{\lower3pt\hbox{$\scriptscriptstyle\frown$}}\over{\varphi}  } }_{U,P}}$ other than the smaller one $\varphi_{U,P}$, which will increase the AoI of the PIoT. The increase of $\eta$ and the decrease of ${{\mathord{\buildrel{\lower3pt\hbox{$\scriptscriptstyle\frown$}}\over{\varphi}  } }_{U,P}}$ at the same time can lead to the decrease or increase of the AoI. For the cases $d_{SP}=30,35$ m, we observe a slight increase around $P_P= 25$ dBm in Fig \ref{montsimuP2}, which could be caused by the large value of $\eta$ as shown in Fig \ref{fig:review}. After $P_P$ goes beyond 25 dBm, the value of $\eta$ starts to decrease as $P_P$ increases. Together with the decrease of the outage probabilities ${{\mathord{\buildrel{\lower3pt\hbox{$\scriptscriptstyle\frown$}}\over{\varphi}  } }_{U,P}}$  and $\varphi_{U,P}$, they jointly lead to a sharp decrease of the AoI of the PIoT.
\begin{figure}
\centering
  {\scalebox{0.5}{\includegraphics {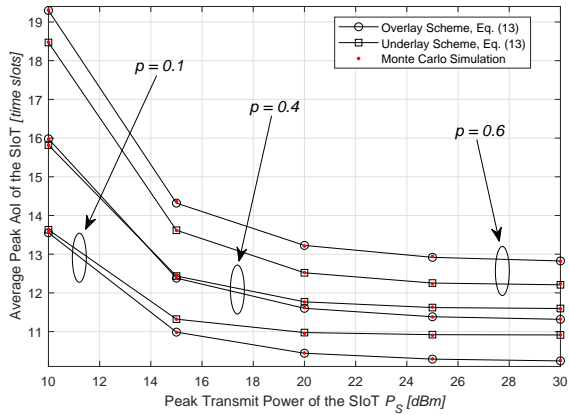}}}
\caption{The average peak AoI of the SIoT versus the peak transmit power of the SIoT for different generation rate $p$, where $P_P=25$ dBm, $d_{SP}=80$ m, $d_{PS}=150$ m, $q = 0.1$ and $I_C = 5N_0$.
\label{montsimuS1}}
\end{figure}
\begin{figure}
\centering
  {\scalebox{0.5}{\includegraphics {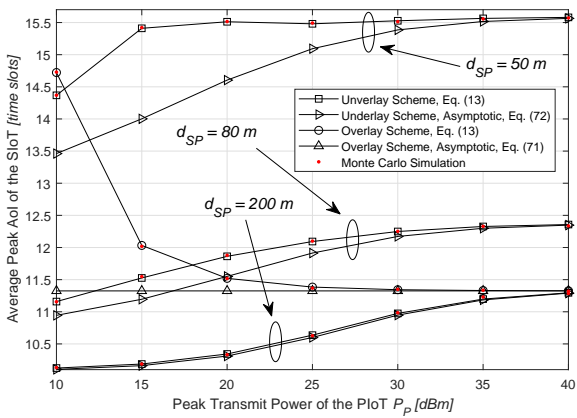}}}
\caption{The average peak AoI of the SIoT versus the peak transmit power of the PIoT for different distance between the SD and the PAP $d_{SP}$, where $P_S=25$ dBm, $d_{PS}=100$ m, $p=0.4$, $q = 0.1$ and $I_C = 5N_0$.
\label{montsimuS2}}
\end{figure}

Fig. \ref{montsimuS1} and Fig. \ref{montsimuS2} show the average peak AoI of the SIoT versus the peak transmit power of the SIoT and the PIoT, respectively. Figs. \ref{montsimuS1}
and \ref{montsimuS2} validate the derived analytical and asymptotic expressions of the average peak AoI of the SIoT. We can see in Fig. \ref{montsimuS2} that the derived asymptotic expressions for the average peak AoI of the SIoT approach their exact counterparts for $P_P=40$ dBm for all the simulated cases. Besides, in Fig. \ref{montsimuS1}, we observe that both the average peak AoI of the overlay and underlay schemes decrease as the peak transmit power constraint $P_S$ grows. The overlay and underlay schemes perform similarly in terms of the average peak AoI of the SIoT when $p=0.4$. Additionally, the overlay scheme outperforms the underlay scheme for $p=0.1$ and the underlay scheme outperforms the overlay scheme for $p=0.6$. This observation indicates that the overlay scheme is superior to the underlay scheme in terms of the average peak AoI of the SIoT when the generation rate of the PIoT is low. This is because the SIoT can easily get access to the spectrum when $p$ is low and the transmit power of the SD in the overlay scheme is not confined by the interference constraint as in the underlay scheme. However, when $p$ is high, the SIoT can hardly access the spectrum and the underlay scheme can overcome this problem. Similar to the analysis given in Fig \ref{montsimuP2} and Fig. \ref{fig:review}, in Fig. \ref{montsimuS2}, the non-convex shape for the average peak AoI of the SIoT exists for small $d_{SP}$, e.g., $d_{SP} = 50$ m. This is because that increasing $P_P$ on one hand increases the outage probability ${{\mathord{\buildrel{\lower3pt\hbox{$\scriptscriptstyle\frown$}}\over{\varphi}  } }_{U,S}}$  of the SIoT, leading to an increase of the average peak AoI. On the other hand, it also changes the probability of the event that the SIoT suffers from interference when it is transmitting (i.e., the PIoT transmits at the same time). The overall change of the AoI for the SIoT is the composition of these effects.
\begin{figure}
\centering
  {\scalebox{0.5}{\includegraphics {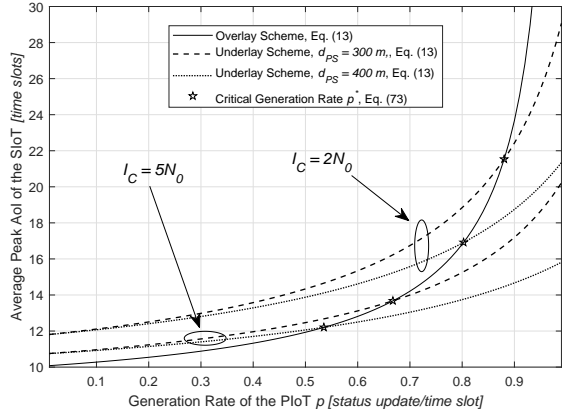}}}
\caption{The average peak AoI of the SIoT versus the generation rate of the PIoT for different system setups, where $P_P=40$ dBm, $P_S = 25$ dBm, $d_{SP}=80$ m and $q = 0.1$.
\label{compare1}}
\end{figure}
\begin{figure}
\centering
  {\scalebox{0.5}{\includegraphics {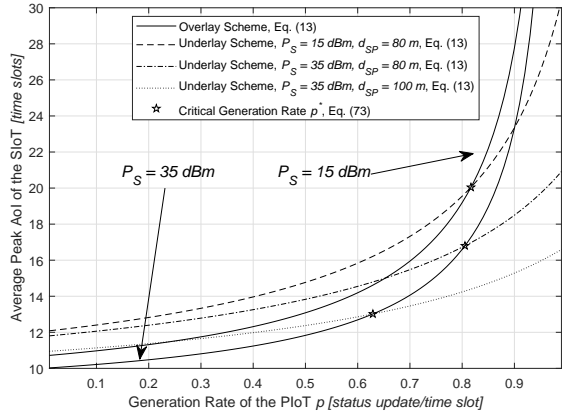}}}
\caption{The average peak AoI of the SIoT versus the generation rate of the PIoT for different system setups, where $P_P=40$ dBm, $d_{PS}=400$ m, $q = 0.1$ and $I_C =2N_0$.
\label{compare2}}
\end{figure}

We now compare the average peak AoI of the SIoT when the PIoT operates at high SNR in Fig. \ref{compare1} and Fig. \ref{compare2}, i.e., $P_P = 40$ dBm. From Figs. \ref{compare1} and \ref{compare2}, we can first observe that there exists a critical generation rate of the PIoT $p^*$ that the average peak AoIs of the SIoT for the overlay and underlay schemes are identical for all the simulated cases. More importantly, when the generation rate of the PIoT $p<p^*$, the overlay scheme outperforms the underlay scheme in terms of the average peak AoI of the SIoT. When $p>p^*$, the underlay scheme is better than the overlay scheme. This observation validates the analysis given in Section V-C. Besides, from Fig. \ref{compare1}, we can see that $p^*$ shifts to the right as $I_C$ and $d_{PS}$ decreases, which indicates that the overlay scheme is suitable to a CR-IoT with a lower interference constraint $I_C$ and a higher interference power between the PD and the SAP. Finally, in Fig. \ref{compare2}, we can find out that the critical value $p^*$ grows as $P_S$ and $d_{SP}$ decreases, which shows that the overlay scheme is preferable for a CR-IoT with a lower secondary peak transmit power constraint and higher interference power between the SD and the PAP.


%
\section{Conclusions}
In this paper, we have analyzed the AoI performance of a CR-IoT consisting of one PIoT and one SIoT. Both the PIoT and the SIoT are assumed to adopt the classical ARQ, in which the PIoT (SIoT) keeps transmitting the current status update until the associated access point decodes the status update correctly or a new status update is generated. The preemption model was considered that the PIoT (SIoT) discards the current status update when a new status update is generated. By adopting the overlay and underlay schemes as the access strategies, we then analyzed both the exact and asymptotic expressions of the average peak AoI for the PIoT (SIoT) in closed-forms. Based on the derived asymptotic expressions, we further obtained a critical value of the generation rate for the PIoT, which determines the superiority between the overlay and the underlay schemes. Numerical results showed that the overlay scheme is more suitable for a system with a lower generation rate at the PIoT, a lower peak transmit power constraint of the SIoT, and a higher interference between the PIoT and the SIoT.
\appendices
\section{Proof of Lemma \ref{lemma1}}\label{app1}
The service of the status update is finished only when it is not preempted by other status updates and finally decoded correctly at the PAP. Since the outage probability for each round of transmission and the generation of new status updates are independent, the probability that the status update is successfully decoded at the PAP after $k$ times of transmission is given by
\begin{equation}\label{app1eq1}
\chi_k = {{{\varphi} _{O,P}}^{k - 1}}{\left( {1 - p} \right)^{k - 1}}\left( {1 - {{\varphi} _{O,P}}} \right).
\end{equation}
The total probability that the status update is not preempted and successfully received by the PAP is given by
\begin{equation}\label{app1eq2}
{\sum\limits_{k = 1}^\infty  {{\chi _k}} }={{1 - {\varphi _{O,P}}} \over {1 - {\varphi _{O,P}} + p{\varphi _{O,P}}}},
\end{equation}
where we use \cite[Eq. (0.231.1)]{Tableofintegral} to simplify the summation term. The PMF of $S_P$ can now be derived as
\begin{equation}
{f_{{S_P}}}\left( k \right)={{{\chi _k}} \over {\sum\limits_{k = 1}^\infty  {{\chi _k}} }},
\end{equation}
which gives the desired result in (\ref{PMFSP}) after some manipulations.
\section{Proof of Lemma \ref{lemma2}}\label{app2}
\begin{figure}[htb]
\centering
 \subfigure[$W_S<W_P$]
  {\scalebox{0.5}{\includegraphics {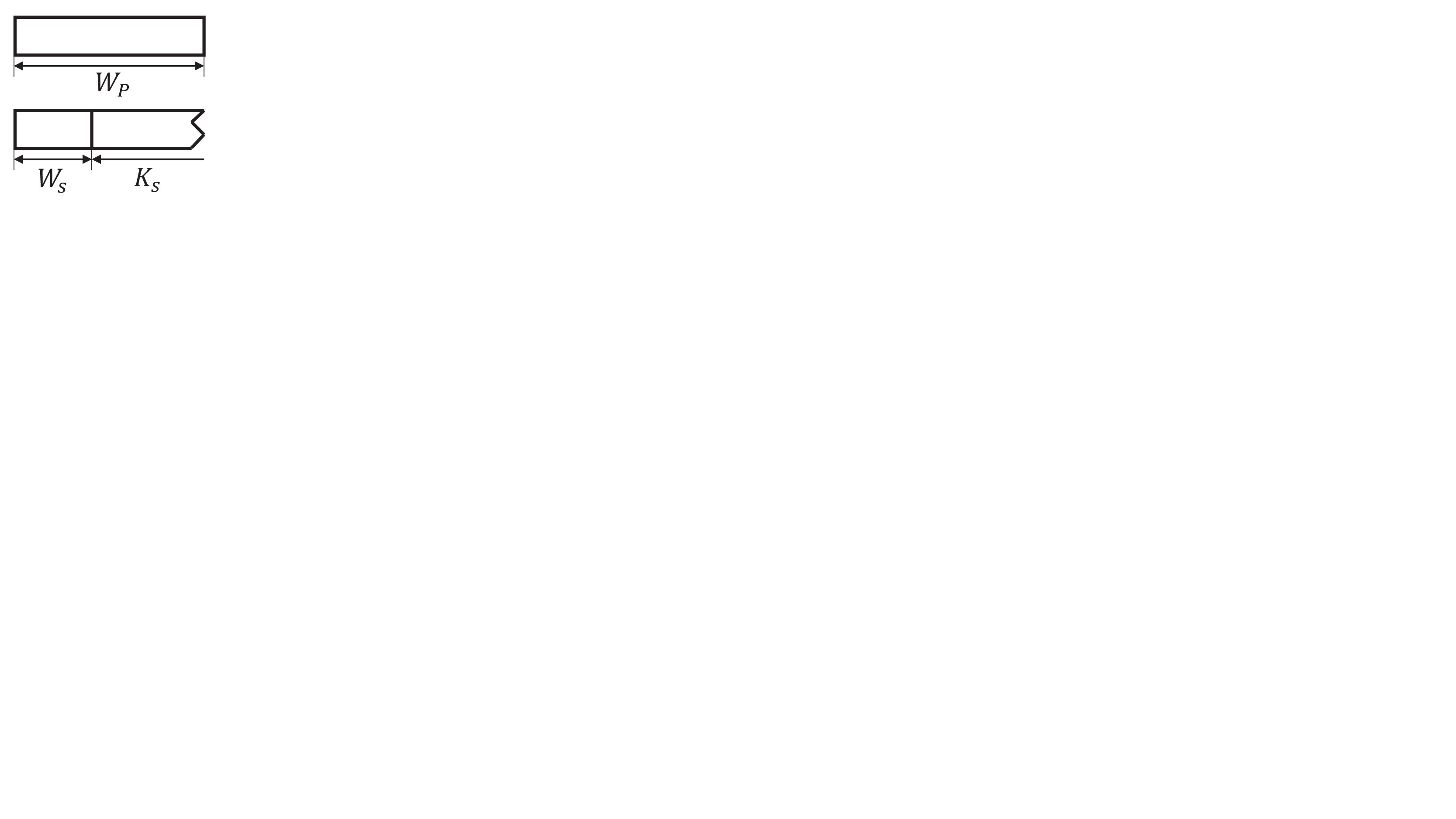}
  \label{fig:WSWP1}}}
\hfil
 \subfigure[$W_S \ge W_P$]
  {\scalebox{0.5}{\includegraphics {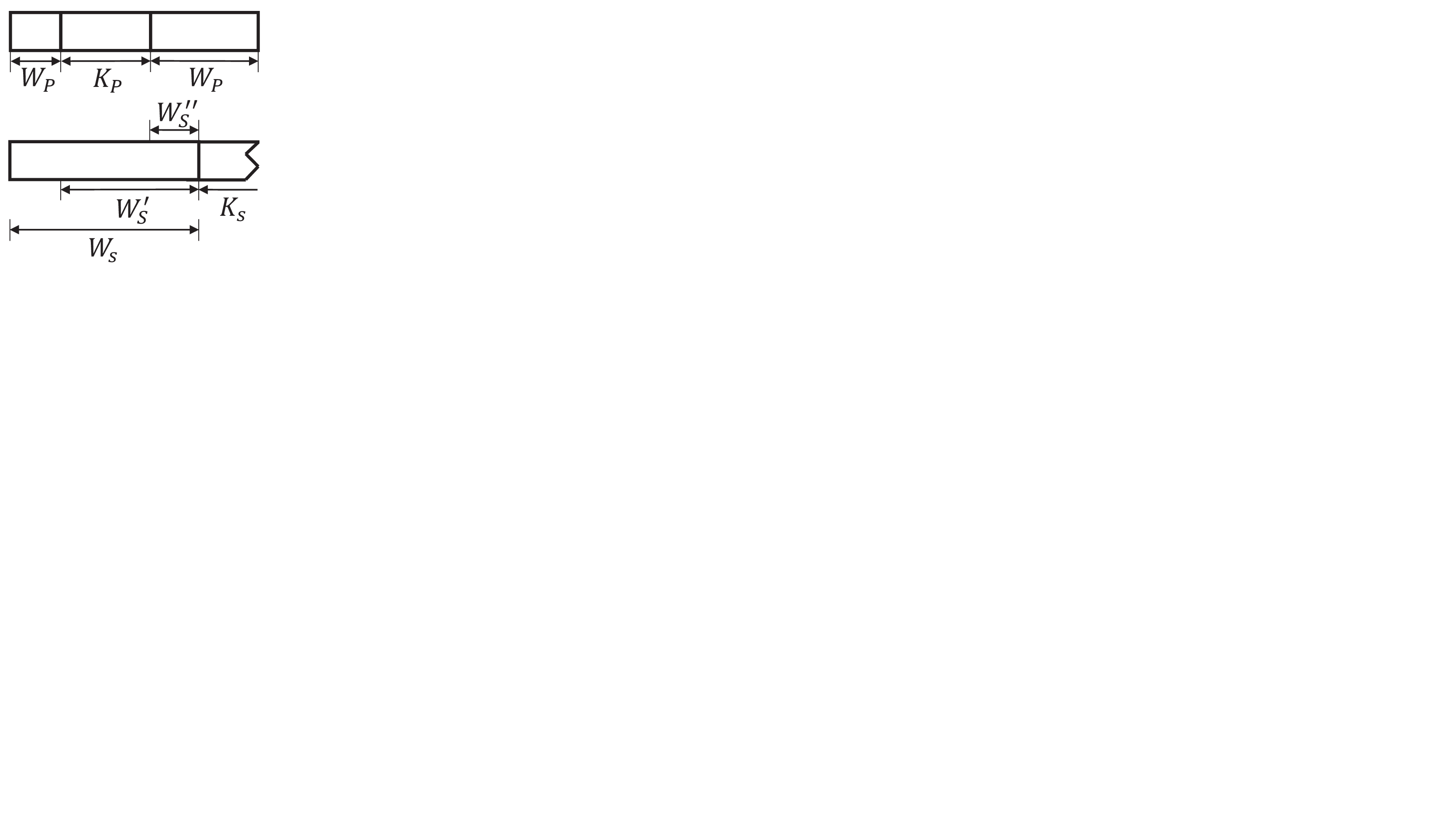}
\label{fig:WSWP2}}}
\caption{The illustration of the event $I_{K_S}$.
\label{WSWP}}
\end{figure}
It is important to notice that when the SIoT finishes the service of a status update and waits for the generation of a new status update, the PIoT is also waiting for a new status update to be generated. This is because that the SD can transmit its status update only when the PD is idle in the overlay scheme. The probability of $I_{K_S}$ can be calculated recursively from the following analysis. As illustrated in Fig. \ref{fig:WSWP1}, $I_{K_S}$ happens when the waiting time of the SIoT is shorter than the PIoT, i.e., $W_S<W_P$, and the initial slot of $K_S$ is thus an idle slot. Besides, as depicted in Fig. \ref{fig:WSWP2}, if the waiting time $W_S \ge W_P$, $I_{K_S}$ may also happen when the waiting time of the SIoT is longer than the waiting period and the transmission period of the PIoT. During the second waiting period of the PIoT, a new status update is generated at the SIoT and the initial slot of $K_S$ is also an idle slot. Recall that $W_P$, $K_P$ and $W_S$ follow the Bernoulli process, and such process is a memoryless process. As depicted in Fig. \ref{fig:WSWP2}, the remaining interval $W_S^\prime$ when $W_S$ is longer than $W_P$, and the remaining interval $W_S^{\prime\prime}$ when $W_S$ is longer than both the $W_P$ and $K_P$, also follow the same geometric distribution as $W_S$. The probability for the case shown in Fig. \ref{fig:WSWP2} can thus be evaluated as $\Pr \left\{ {{W_S} \ge {W_P}} \right\}\Pr \left\{ {{W_S} \ge {K_P}} \right\}\Pr \left\{ {{W_S} < {W_P}} \right\}$.

Fig. \ref{fig:WSWP1} and \ref{fig:WSWP2} indicate the event $I_{K_S}$ for the cases that the PIoT has no transmitting period, and one transmitting period during the waiting time $W_S$, respectively. In general, the PIoT may experience an arbitrary number of transmitting periods during $W_S$ and $I_{K_S}$ still happens, from the above analysis, the probability term $\Pr\left\{I_{K_S}\right\}$ can be derived as
\begin{align}
&\Pr\left\{I_{K_S}\right\} \nonumber\\
&=\Pr \left\{ {{W_S} < {W_P}} \right\}\sum\limits_{n = 0}^\infty  {{{\left( {\Pr \left\{ {{W_S} \ge {K_P}} \right\}\Pr \left\{ {{W_S} \ge {W_P}} \right\}} \right)}^n}}\nonumber\\
&= {{\Pr \left\{ {{W_S} < {W_P}} \right\}} \over {1 - \Pr \left\{ {{W_S} \ge {K_P}} \right\}\Pr \left\{ {{W_S} \ge {W_P}} \right\}}}.\label{IK1}
\end{align}
With the PMFs of $W_P$, $K_P$ and $W_S$ given in (\ref{PMFWP}), (\ref{PMFKP}) and (\ref{PMFWS}), we can derive that $\Pr \left\{ {{W_S} < {W_P}} \right\} = {{q - pq} \over {q + p - pq}}$, $\Pr \left\{ {{W_S} \ge {W_P}} \right\} = {p \over {q + p - pq}}$ and $\Pr \left\{ {{W_S} \ge {K_P}} \right\} = {{\left( {1 - {\varphi _{O,P}}} \right)\left( {1 - q} \right)} \over {1 - {\varphi _{O,P}} + q{\varphi _{O,P}}}}$. The desired expression of $\Pr\left\{I_{K_S}\right\}$ in (\ref{IK}) can be obtained by substituting these results into (\ref{IK1}), and the term $\Pr\left\{B_{K_S}\right\}$ can be obtained by $\Pr\left\{B_{K_S}\right\}=1-\Pr\left\{I_{K_S}\right\}$.
\section{Proof of Lemma \ref{lemma4}}\label{app3}
When the PIoT operates at high SNR and ${\varphi_{U,P}},{{\mathord{\buildrel{\lower3pt\hbox{$\scriptscriptstyle\frown$}}\over
 {\varphi}  } }_{U,P}} \to 0$, the structure of the transition matrix in (\ref{MCtransition}) can be simplified. By solving the linear equations with four variables $\pmb \pi \pmb {M}= \pmb \pi $ and let $\phi=\left( {1 - p} \right)\left( {1 - q} \right)\left( {1 - {\varphi _{U,S}}} \right) + p\left( {1 - q} \right)\left( {1 - {{{{\mathord{\buildrel{\lower3pt\hbox{$\scriptscriptstyle\frown$}}\over
 \varphi } }}}_{U,S}}} \right)$ for notation simplicity, the stationary distribution can be derived in an explicit form as ${\pi _1}= {{\left(1-p\right)\phi} \over {q + \phi}}$, ${\pi _2} = {{p\phi} \over {q + \phi}}$, ${\pi _3} = {{\left(1-p\right)q} \over {q +\phi}}$, and ${\pi _4} = {{pq} \over {q + \phi}}$. Substituting theses results into (\ref{YSU}), the expectation of $Y_S$ is simplified to
 \begin{equation}\label{YSasym}
\mathbb{E}\left[Y_S\right]\approx{{1 - q} \over q}+{1 \over {\left( {1 - p} \right)\left( {1 - {\varphi _{U,S}}} \right) + p\left( {1 - {{{{\mathord{\buildrel{\lower3pt\hbox{$\scriptscriptstyle\frown$}}\over
 \varphi } }}}_{U,S}}} \right)}}.
\end{equation}
Substituting ${\varphi_{U,P}},{{\mathord{\buildrel{\lower3pt\hbox{$\scriptscriptstyle\frown$}}\over
 {\varphi}  } }_{U,P}}\to0$ into (\ref{SZIU}) and (\ref{SZBU}), the conditional expectations of the service time of the SIoT can be simplified to
\begin{equation}
\mathbb E \left[ {S_S \left| I_{S_S} \right.} \right]\approx{{1 - {{{{\mathord{\buildrel{\lower3pt\hbox{$\scriptscriptstyle\frown$}}\over
 \varphi } }}}_{U,S}}p\left( {1 - q} \right) + {\varphi _{U,S}}p\left( {1 - q} \right){\Pr \left\{ {\Phi \left| B_S \right.} \right\} \over \Pr \left\{ {\Phi \left| I_S \right.} \right\}}} \over \phi },
\end{equation}
\begin{equation}
\begin{split}
&\quad \mathbb E \left[ {S_S \left| B_{S_S} \right.} \right]\\
&\approx{{1 - {\varphi _{U,S}}\left( {1 - p} \right)\left( {1 - q} \right) + {{{{\mathord{\buildrel{\lower3pt\hbox{$\scriptscriptstyle\frown$}}\over
 \varphi } }}}_{U,S}}\left( {1 - p} \right)\left( {1 - q} \right){\Pr \left\{ {\Phi \left| I_S \right.} \right\} \over \Pr \left\{ {\Phi \left| B_S \right.} \right\}}} \over \phi }.
\end{split}
\end{equation}
It can be verified that when the PIoT operates at high SNR, the probability that new status updates of the SIoT is generated in an idle slot or in a busy slot is given by $\Pr \left\{ {{I_S}} \right\}=1-p$ and $\Pr \left\{ {{B_S}} \right\}=p$. According to (\ref{SZ}), (\ref{ISZ}) and (\ref{BSZ}), the expectation of $S_S$ when the PIoT operates at high SNR can thus be derived as
\begin{equation}\label{SSasym}
\begin{split}
\mathbb E \left[ {{S_S}} \right]&\approx {{\left(1-p\right)\Pr \left\{ {\Phi \left| {{I_Z}} \right.} \right\}\mathbb E \left[ {S_S \left| I_{S_S} \right.} \right]+p\Pr \left\{ {\Phi \left| {{B_Z}} \right.} \right\}} \over {\left(1-p\right)\Pr \left\{ {\Phi \left| {{I_Z}} \right.} \right\} + p\Pr \left\{ {\Phi \left| {{B_Z}} \right.} \right\}}}\\
&\quad + {{p\Pr \left\{ {\Phi \left| {{B_Z}} \right.} \right\}\mathbb E \left[ {S_S \left| B_{S_S} \right.} \right]} \over {\left(1-p\right)\Pr \left\{ {\Phi \left| {{I_Z}} \right.} \right\} + p\Pr \left\{ {\Phi \left| {{B_Z}} \right.} \right\}}}\\
 &\approx{1 \over {q + \left( {1 - q} \right)\left[{\left( {1 - p} \right)\left( {1 - {\varphi _{U,S}}} \right) + p\left( {1 - {{{{\mathord{\buildrel{\lower3pt\hbox{$\scriptscriptstyle\frown$}}\over
 \varphi } }}}_{U,S}}} \right)}\right]}}.
\end{split}
\end{equation}
Substituting (\ref{YSasym}) and (\ref{SSasym}) into (\ref{AoIexpression}), we obtain the desired result for the SIoT given in (\ref{SUasym}) for the case that the PIoT operates at high SNR.
\ifCLASSOPTIONcaptionsoff
  \newpage
\fi

\bibliographystyle{IEEEtran}
\bibliography{References}
\begin{IEEEbiography} [{\includegraphics[width=1in,height=1.25in,clip,keepaspectratio]{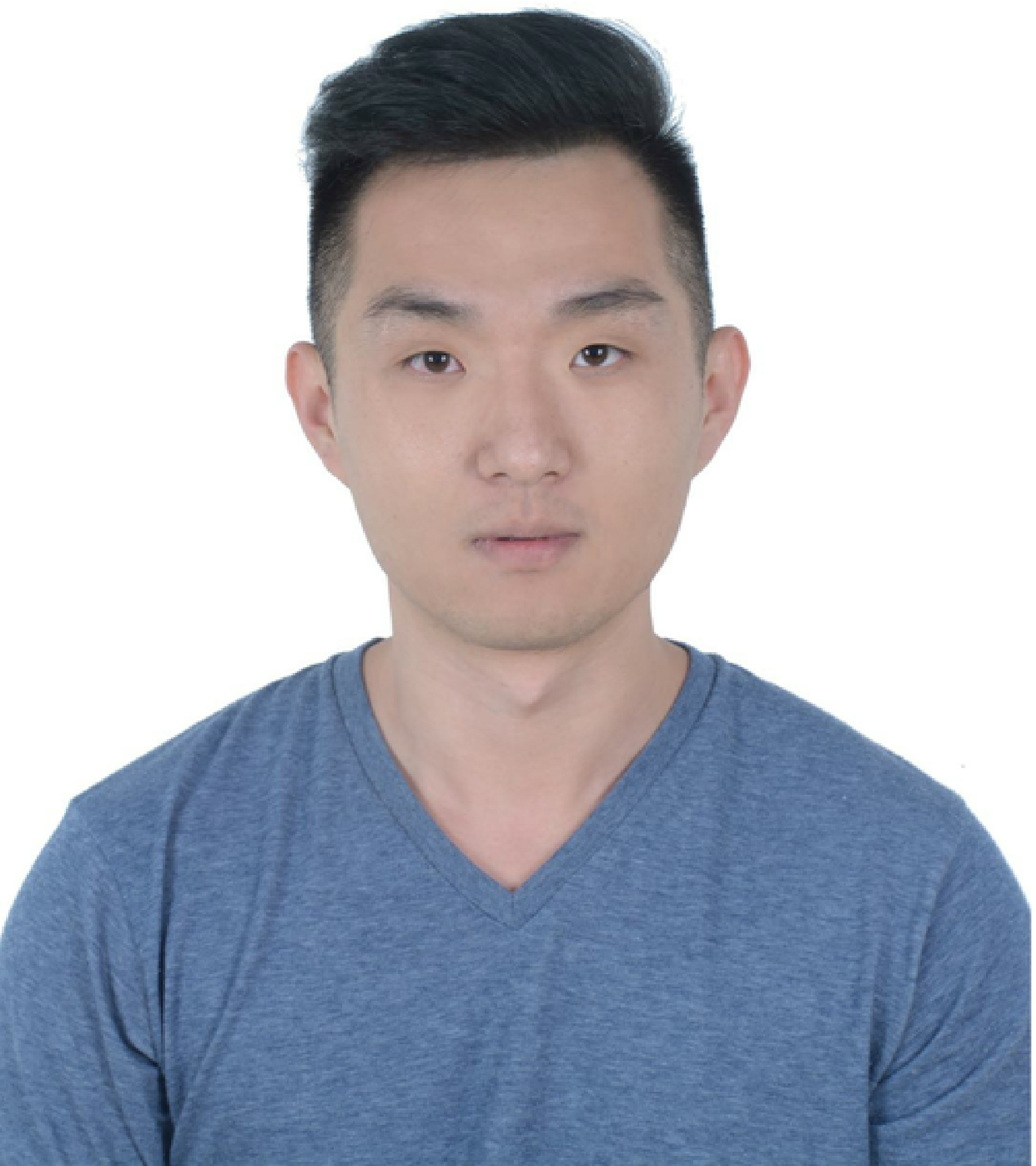}}]
{Yifan Gu} received the B.S. and Ph. D. degrees in Electrical Engineering from The University of Sydney in 2015 and 2019, respectively. He is currently a postdoctoral researcher at the University of Sydney. His research with Prof. Yonghui Li and Prof. Branka Vucetic involves relaying communication, age of information (AoI), industrial Internet of Things (IoT), and industrial wireless communication.
\end{IEEEbiography}

\begin{IEEEbiography} [{\includegraphics[width=1in,height=1.25in,clip,keepaspectratio]{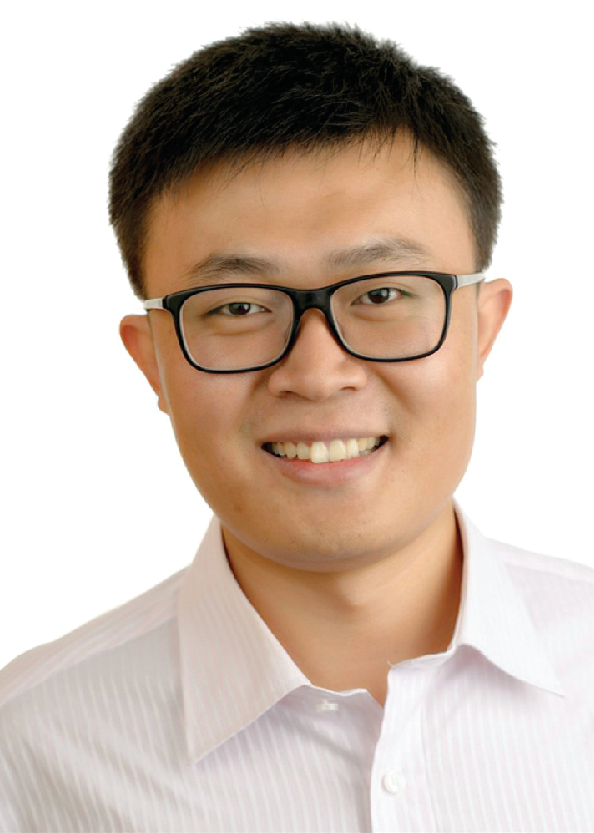}}]
{He (Henry) Chen} (S'10-M'16) received a Ph.D. degree in Electrical Engineering from The University of Sydney, Sydney, Australia, in 2015. He was a Research Fellow with the School of Electrical and Information Engineering, The University of Sydney before he joined the Department of Information Engineering at the Chinese University of Hong Kong as a faculty member, where he is now a Research Assistant Professor.

Dr. Chen is current research interests are in the field of industrial Internet of Things, with a particular focus on ultra-reliable low-latency wireless for industrial automation, timely status update in real-time monitoring, and time-sensitive networking (TSN).
\end{IEEEbiography}

\begin{IEEEbiography} [{\includegraphics[width=1in,height=1.25in,clip,keepaspectratio]{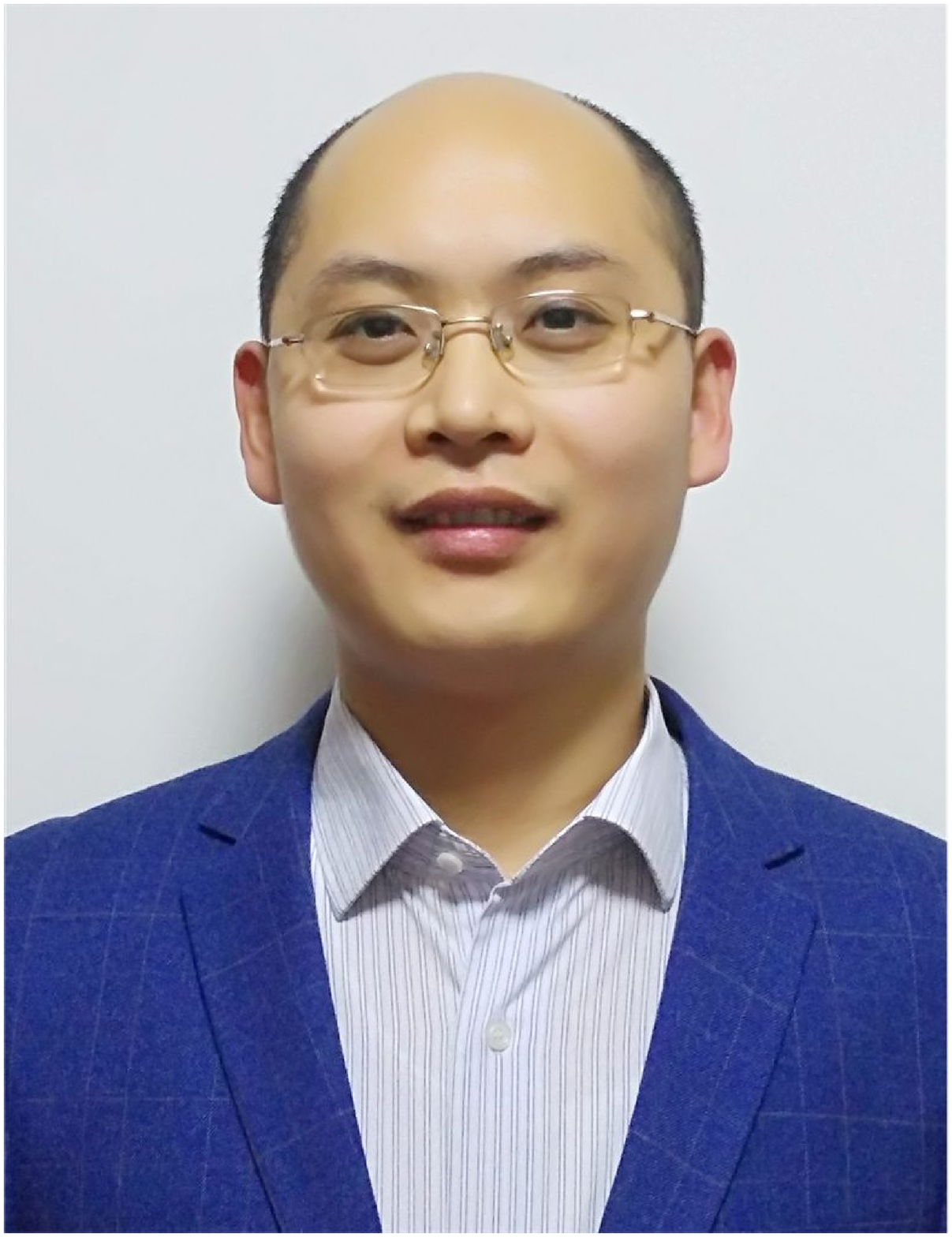}}]
{Chao Zhai} (S'12-M'15) received the B.Sc. degree in communication engineering and the M.Sc. degree in communication and information systems from Shandong University, Jinan, China, in 2007 and 2010, respectively, and the Ph.D. degree in electrical engineering from the University of New South Wales, Sydney, Australia, in 2013. From 2011 to 2012, he was a Visiting Ph.D. Student with The Chinese University of Hong Kong. In 2013, he was a Visiting Ph.D. student with The Hong Kong University of Science and Technology. From 2014 to 2016, he held a post-doctoral position at Shandong University, where he is currently with the School of Information Science and Engineering, Shandong University, as an Associate Researcher. His research interests include high-efficient transmissions in heterogeneous networks, massive MIMO, spectrum sharing in cognitive radio networks, and energy harvesting wireless communications.
\end{IEEEbiography}

\begin{IEEEbiography} [{\includegraphics[width=1in,height=1.25in,clip,keepaspectratio]{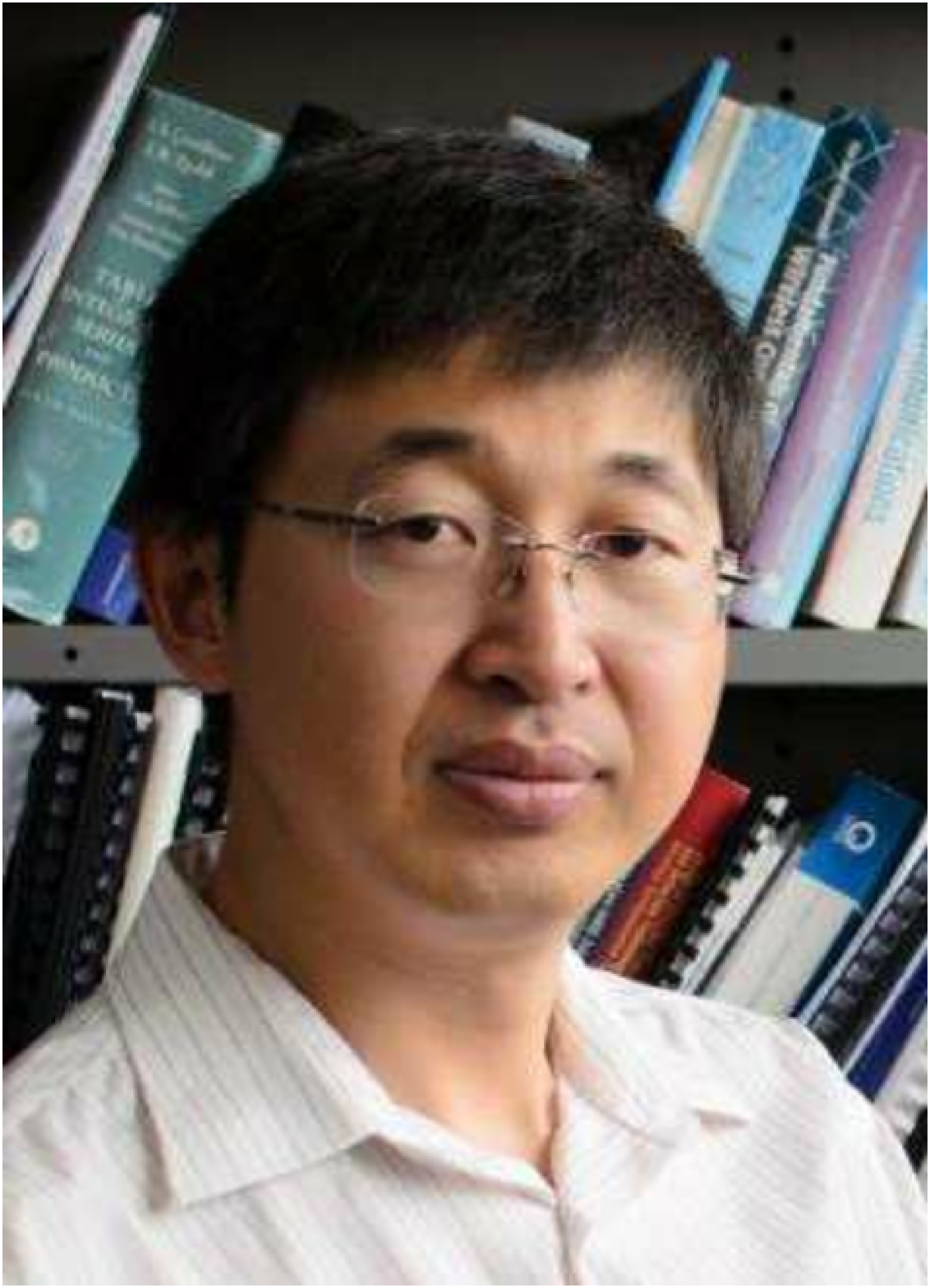}}]
{Yonghui Li} (M'04-SM'09-F'19) received his PhD degree in November 2002 from Beijing University of Aeronautics and Astronautics. From 1999 to 2003, he was affiliated with Linkair Communication Inc, where he held a position of project manager with responsibility for the design of physical layer solutions for the LAS-CDMA system. Since 2003, he has been with the Centre of Excellence in Telecommunications, the University of Sydney, Australia. He is now a Professor in School of Electrical and Information Engineering, University of Sydney. He is the recipient of the Australian Queen Elizabeth II Fellowship in 2008 and the Australian Future Fellowship in 2012.

His current research interests are in the area of wireless communications, with a particular focus on MIMO, millimeter wave communications, machine to machine communications, coding techniques and cooperative communications. He holds a number of patents granted and pending in these fields. He is now an editor for IEEE transactions on communications and IEEE transactions on vehicular technology. He was also the guest editor for IEEE JSAC Special issue on Millimeter Wave Communications for Future Mobile Networks. He received the best paper awards from IEEE International Conference on Communications (ICC) 2014, IEEE PIMRC 2017, and IEEE Wireless Days Conferences (WD) 2014. He is Fellow of IEEE.
\end{IEEEbiography}

\begin{IEEEbiography} [{\includegraphics[width=1in,height=1.25in,clip,keepaspectratio]{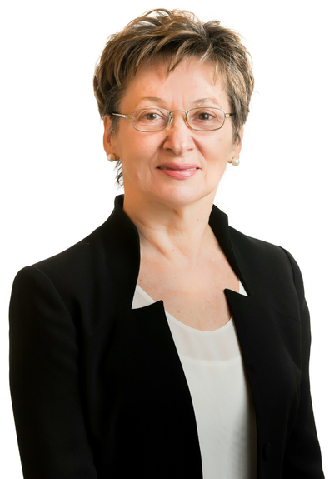}}]
{Branka Vucetic} (SM'00-F'03) is an ARC Laureate Fellow and Professor of Telecommunications, Director of the Centre of Excellence in Telecommunications at the University of Sydney. During her career she has held research and academic positions in Yugoslavia, Australia, UK and China. Her research interests include coding, communication theory and signal processing and their applications in wireless networks and industrial internet of things. Prof Vucetic co-authored four books and more than four hundred papers in telecommunications journals and conference proceedings. She is a Fellow of the Australian Academy of Technological Sciences and Engineering and a Fellow of the IEEE.
\end{IEEEbiography}
\end{document}